%% file: main.tex
\documentclass{article}[11pt]
\usepackage[margin=1in]{geometry}

\usepackage{libertine}
\usepackage{slantsc}
\usepackage{pifont}
\usepackage{microtype}
\usepackage[T1]{fontenc}
\usepackage[utf8]{inputenc}
\usepackage[american]{babel}
\usepackage[scr=rsfso]{mathalfa}

\usepackage{graphicx}
\usepackage{bbm}
\usepackage{amsmath,amsthm,amssymb}
\usepackage{bm}
\usepackage{breakcites}
\usepackage{booktabs}
\usepackage{mathtools}
\usepackage{calc}
\usepackage{soul}

\newtheorem{lemma}{Lemma}
\newtheorem{claim}{Claim}
\newtheorem{theorem}{Theorem}
\newtheorem{corollary}{Corollary}

\newtheorem{problem}{Problem}
\newtheorem{fact}{Fact}

\newtheorem*{rep@theorem}{\rep@title}
\newcommand{\newreptheorem}[2]{%
\newenvironment{rep#1}[1]{%
 \def\rep@title{#2 \ref{##1}}%
 \begin{rep@theorem}}%
 {\end{rep@theorem}}}
\makeatother
\newreptheorem{theorem}{Theorem}
\newreptheorem{corollary}{Corollary}
\newreptheorem{definition}{Definition}

\theoremstyle{definition}
\newtheorem{definition}{Definition}

\usepackage[labelfont=bf,textfont=it,font=small]{caption}

\usepackage{algorithm}
\usepackage{algorithmicx}
\usepackage[noend]{algpseudocode}
\algrenewcommand\algorithmiccomment[1]{\hfill{\color{gray}$\triangleright$~#1}}
\algrenewcommand\algorithmicrequire{\textbf{Input:}}
\algrenewcommand\algorithmicensure{\textbf{Output:}}
\newcommand{\algrule}[1][.2pt]{\par\vskip.2\baselineskip\hrule height #1\par\vskip.2\baselineskip}

\usepackage{dsfont}
\newcommand{\ones}{\mathds{1}}
\newcommand{\E}{\mathbb{E}}
\newcommand{\Var}{\mathrm{Var}}
\newcommand{\poly}{\mathrm{poly}}

\usepackage{xcolor}
\definecolor{c0}{HTML}{641a80}
\definecolor{c1}{HTML}{b73779}

\usepackage{comment}

\usepackage{enumerate}

\usepackage{hyperref}
\hypersetup{
    colorlinks,
    linkcolor=c0,
    citecolor=c1,
    urlcolor=c0
}
\usepackage[capitalize,nameinlink,noabbrev]{cleveref}
\crefformat{equation}{#2(#1)#3}
\crefmultiformat{equation}{(#2#1#3)}%
{ and~#2(#1)#3)}{, #2(#1)#3}{ and~#2(#1)#3}
\crefname{assumption}{Assumption}{Assumptions}
\crefname{problem}{Problem}{Problems}
\crefname{claim}{Claim}{Claims}
\crefname{fact}{Fact}{Facts}

\newcommand{\norm}[1]{\|#1\|}

\newcommand{\bv}[1]{\boldsymbol{\mathbf{#1}}}

\newcommand{\EE}{\operatorname{\mathbb{E}}}

\newcommand{\R}{\mathbb{R}}

\newcommand{\T}{\mathsf{T}}
\newcommand{\Fro}{\mathsf{F}}
\newcommand{\F}{\mathcal{F}}

\newcommand{\OPT}{\mathrm{OPT}}
\newcommand{\tr}{\operatorname{tr}}

\newcommand{\transpose}{\mathsf{T}}
\newcommand{\family}{\mathcal{F}}
\newcommand{\logfamily}{\log |\mathcal{F}|}

\usepackage{nicematrix}

\usepackage[color=blue!10,linecolor=blue,textsize=footnotesize]{todonotes}

\allowdisplaybreaks

\title{Query Efficient Structured Matrix Learning}

\author{
    Noah Amsel\thanks{New York University (\url{noah.amsel@nyu.edu}, \url{pratyushavi@nyu.edu}, \url{tyler.chen@nyu.edu}, \url{fd2135@nyu.edu}, \url{chinmay.h@nyu.edu}, \url{cmusco@nyu.edu})} \and
    Pratyush Avi\footnotemark[1]\and
    Tyler Chen\footnotemark[1]\and
    Feyza Duman Keles\footnotemark[1]\and
    Chinmay Hegde\footnotemark[1]\and
    Cameron Musco\thanks{University of Massachusetts Amherst (\url{cmusco@cs.umass.edu})}\and
    Christopher Musco\footnotemark[1]\and
    David Persson\footnotemark[1]\,\thanks{New York University \& Flatiron Institute (\url{dpersson@flatironinstitute.org})} 
}

\DeclareMathOperator*{\argmin}{arg\,min}

\date{}

\input{bib_macros}

\begin{document}

\maketitle

\begin{abstract}
    We study the problem of learning a structured approximation (low-rank, sparse, banded, etc.) to an unknown matrix $\bv{A}$ given access to matrix-vector product (matvec) queries of the form $\bv{x} \rightarrow \bv{A}\bv{x}$ and $\bv{x} \rightarrow \bv{A}^\transpose \bv{x}$. This problem is of central importance to algorithms across scientific computing and machine learning, with applications to fast multiplication and inversion for structured matrices, building preconditioners for first-order optimization, and as a model for differential operator learning. Prior work focuses on obtaining query complexity upper and lower bounds for learning specific structured matrix families that commonly arise in applications. 
    
    We initiate the study of the problem in greater generality, aiming to understand the query complexity of learning approximations from general matrix families. Our main result focuses on finding a near-optimal approximation to $\bv{A}$ from any \emph{finite-sized} family of matrices, $\mathcal{F}$. Standard results from matrix sketching show that $O(\log|\mathcal{F}|)$ matvec queries suffice in this setting. This bound can also be achieved, and is optimal, for vector-matrix-vector queries of the form $\bv{x},\bv{y}\rightarrow \bv{x}^\transpose\bv{A}\bv{y}$, which have been widely studied in work on rank-$1$ matrix sensing.
    
    Surprisingly, we show that, in the matvec model, it is possible to obtain a nearly quadratic improvement in complexity, to $\tilde{O}(\sqrt{\log|\mathcal{F}|})$. Further, we prove that this bound is tight up to log-log factors.

    Via covering number arguments, our result extends to well-studied infinite families. As an example, we establish that a near-optimal approximation from any \emph{linear matrix family} of dimension $q$ can be learned with $\tilde{O}(\sqrt{q})$ matvec queries, improving on an $O(q)$ bound achievable via sketching techniques and vector-matrix-vector queries.
\end{abstract}

\vspace{1em}
\begin{center}
\fbox{
\begin{minipage}{.8\textwidth}
    \textbf{Update:}
    Since obtaining the results in this paper, the mathematical capabilities of public large language models have been significantly expanded. Standard publicly accessible LLMs (e.g. from OpenAI, Anthropic, and Google) are all able to improve our main $(3+\epsilon)$ approximation result for finite families to $(1+\epsilon)$ by bootstrapping off of the constant factor result. 
    We have made one such proof available \href{https://github.com/PratyushAvi/Structured-Matrix-Learning/blob/main/1\%2Beps-finite-family-approximation.pdf}{here\footnote{https://github.com/PratyushAvi/Structured-Matrix-Learning/blob/main/1\%2Beps-finite-family-approximation.pdf}}. This proof is a living document. As of August 2026 it has been human verified but not polished in any way. We may polish or add exposition at some later point.
\end{minipage}
}    
\end{center}

\thispagestyle{empty}
\newpage
\setcounter{page}{1}
\section{Introduction}
What can we learn about an unknown matrix $\bv{A}$ from multiplying it by a few vectors? This basic question has been central to numerical linear algebra for decades, underlying work on Krylov subspace methods \cite{Saad:2003,Saad:2011}, matrix sketching \cite{Woodruff:2014,DrineasMahoney:2016}, randomized numerical linear algebra \cite{WeisseWelleinAlvermann:2006,MartinssonTropp:2020}, and more. In its most general form, the task is to interact with an $n\times n$ matrix $\bv{A}$ in an adaptive way, issuing matrix-vector product queries (``matvec queries'') of the form:
\begin{align*}
    \bv{x} \rightarrow \bv{A}\bv{x} \hspace{2em} \text{or} \hspace{2em} \bv{x} \rightarrow \bv{A}^\transpose\bv{x},
\end{align*}
where $\bv{x}\in \R^{n}$ can be chosen to be any vector. We consider the setting where we can issue multiple matvec queries $\bv{x}_1, \ldots, \bv{x}_m$, and each $\bv{x}_{i}$ can be chosen {adaptively} depending on the results of the previous queries, $\bv{x}_1, \ldots, \bv{x}_{i-1}$.

Understanding the matvec query complexity of linear algebraic problems is critical in applications where $\bv{A}$ is only accessible as an implicit operator (algorithms that work in this setting are called ``matrix-free'') or when the goal is to avoid explicit operations on the entries of $\bv{A}$ to take computational advantage of the natural parallelism of matrix-vector multiplication. In recent years, we have seen a flurry of success in proving better upper and lower bounds (often matching) on the matvec query complexity of central problems like linear system solving \cite{BravermanHazanSimchowitz:2020}, trace estimation \cite{MeyerMuscoMusco:2021,WoodruffZhangZhang:2022,EpperlyTroppWebber:2024,ChenChenLi:2024}, eigenvalue estimation \cite{SimchowitzEl-AlaouiRecht:2018,ChenTrogdonUbaru:2021,BravermanKrishnanMusco:2022,MuscoMuscoRosenblatt:2025,BhattacharjeeJayaramMusco:2025}, matrix function computation \cite{ChewiDios-PontLi:2024}, and much more \cite{SunWoodruffYang:2019,MeyerSwartworthWoodruff:2025}. 

\subsection{Matrix Approximation}
One of the most important problems studied in the matvec query model is that of learning a near-optimal approximation to $\bv{A}$ from a given ``hypothesis class'' of structured matrices, $\family$. For example, $\mathcal{F}$ could be all rank-$k$ matrices or all tridiagonal matrices. We are interested in this problem in the \emph{agnostic learning} setting, where $\bv{A}$ \emph{is not} assumed to lie in $\family$, nor to be generated from some ``true'' matrix in $\mathcal{F}$, e.g., by adding i.i.d. noise. Formally, we consider:
\begin{problem}
\label{prob:intro}
            For an unknown $\bv{A} \in \R^{n \times n}$ and a matrix family $\family \subset \R^{n \times n}$ (a hypothesis class), find $\bv{\tilde{B}}\in \mathcal{F}$ satisfying
            \begin{align*}
                \|\bv{A} - \bv{\tilde{B}}\|_{\Fro} \leq \gamma \cdot \inf_{\bv{B}\in \family}\|\bv{A} - \bv{B}\|_{\Fro}
            \end{align*}
            for some approximation factor $\gamma \geq 1$, using as few matvec queries to $\bv{A}$ and $\bv{A}^\transpose$ as possible.\footnote{We assume $\bv{A}$ is square for simplicity of notation. This is without loss of generality, since the rectangular case can always be handled by padding the smaller dimension with all zeros.}
\end{problem}
In other words, our goal is to find an approximation, $\bv{\tilde{B}}$, from the class $\mathcal{F}$ that is competitive with the \emph{best possible} approximation to $\bv{A}$ in $\family$.
We measure error with respect to the Frobenius norm, which is a popular choice. However, other matrix norms have also been considered \cite{KressnerMasseiRobol:2019,MuscoMusco:2015}. The problem is interesting in several possible regimes for $\gamma$:  when $\gamma$ is close to $1$ (i.e., $\gamma = 1+\epsilon$ for some small $\epsilon$), when $\gamma$ is a constant, or even when $\gamma$ is allowed to grow slowly in $n$ or other problem parameters \cite{AmselChenHalikias:2025}.

\Cref{prob:intro} arises in dozens of applications where the goal is to efficiently obtain an explicit surrogate for an implicit operator. Examples include when $\bv{A}$ is the Hessian matrix of a large-scale optimization problem, for which matvecs can be computed much more efficiently than realizing the matrix \cite{Pearlmutter:1994,DauphinVriesBengio:2015}, when $\bv{A}$ is the inverse of another matrix and we wish to develop a fast solver \cite{LinLuYing:2011,Martinsson:2016}, and when $\bv{A}$ is the solution operator of a differential equation \cite{BoulleTownsend:2023}. For these applications, prior work studies query complexity for a wide variety of structured families, ranging from hierarchically structured matrices, to sparse matrices, to butterfly matrices, and more. 

\Cref{prob:intro} is also relevant when the goal is to replace an already explicit $\bv{A}$ with a structured surrogate that can be stored or operated on more efficiently. 
The matvec query model is relevant in this setting because, on modern hardware, matvecs are an inherently efficient way to access $\bv{A}$. 
A familiar setting is when $\family$ is the class of rank-$k$ matrices for some parameter $k$. For this family, \Cref{prob:intro} captures the problem of finding a near-optimal rank-$k$ approximation, for which the matvec query complexity has been studied extensively \cite{MuscoMusco:2015,BakshiClarksonWoodruff:2022,BakshiNarayanan:2023,MeyerMuscoMusco:2024}.

We provide a more detailed summary of the large and growing body of related work on structured matrix approximation in the matvec query model in \Cref{sec:related_work}. Importantly, essentially all prior work on the problem centers on understanding the query complexity of \emph{specific} structured matrix families. 

\paragraph{Beyond Specific Matrix Families.}
It is natural to ask if a more general theory exists for the query complexity of approximation from a given hypothesis class, $\mathcal{F}$. In supervised learning, a rich theory based on VC dimension, Pollard pseudodimension, fat-shattering dimension, etc. helps categorize the complexity of learning functions from arbitrary hypothesis classes \cite{Pollard:1990,AlonBen-DavidCesa-Bianchi:1997,AnthonyBartlett:1999,Vapnik:2000}. We ask:

\begin{center}\textit{Does a general theory exist to characterize the matvec query complexity of structured matrix approximation?} \end{center}

This question is interesting because \Cref{prob:intro} differs in important ways from standard supervised learning. A first observation is that data collection can be active and adaptive, although there has been significant work on active learning \cite{Settles:2009}. More interesting is the fact that the output corresponding to any query/data example $\bv{x}\in \R^n$ is itself an $n$ dimensional vector, $\bv{A}\bv{x}$ (or $\bv{A}^\transpose\bv{x}$). We might therefore hope to learn \emph{more}  from each example than in a setting with scalar-valued outputs (informally, up to $n$ times more). One such setting with scalar outputs is the well-studied rank-$1$ matrix sensing problem\footnote{This model is called ``rank-$1$'' sensing because $\bv{x}^\transpose\bv{A}\bv{y}$ is equivalent to the Frobenius inner product between $\bv{A}$ and the rank-$1$ matrix $\bv{x}\bv{y}^\transpose$.} from compressed sensing, in which the goal is to find an approximation to $\bv{A}$, but the matrix is only accessible via queries of the form $\bv{x}^\transpose\bv{A}\bv{y}$ \cite{ZhongJainDhillon:2015,DasarathyShahBhaskar:2015,KuengRauhutTerstiege:2017}. From here on, we refer to this setting as the ``vector-matrix-vector'' query model.

Indeed, for many of the families mentioned above, there is a lot to be gained from the multidimensional output. Rank-$k$ matrices, for example, require $O(nk)$ parameters to represent, and $\Omega(nk)$ vector-matrix-vector queries to learn \cite{CandesPlan:2011}, but can be learned (up to constant accuracy) using just $O(k)$ matvec queries \cite{BakshiClarksonWoodruff:2022}. Similarly, matrices with $s$ non-zeros per row have $O(ns)$ parameters and require $\Omega(ns)$ vector-matrix-vector queries to learn \cite{DasarathyShahBhaskar:2015}, but just $O(s)$ matvec queries \cite{AmselChenHalikias:2024}.
In both of these examples, matvec queries allow for a dramatic $n$-factor improvement over scalar queries. Similar improvements are known for other classes \cite{ChenHalikiasKeles:2025,AmselChenHalikias:2025}. What sort of improvement is possible in general?

\subsection{Our Contributions}\label{sec:contributions}
In this work, we take a step towards understanding the power of matvec queries for structured matrix-approximation in greater generality by considering arbitrary \emph{finite families} of matrices (i.e., with finite size $|\mathcal{F}|)$. Such families require $\log(|\mathcal{F}|)$ bits (i.e., parameters) to represent, and it can be shown using standard techniques from matrix sketching that it is always possible to solve \Cref{prob:intro} with constant approximation factor $\gamma$ using $O(\log|\mathcal{F}|)$ matvec queries (see \Cref{sec:one_sided} for more details). Indeed, this bound can also be achieved in the more restrictive vector-matrix-vector (i.e., rank-$1$ matrix-sensing) setting, with queries of the form $\bv{x}^\transpose\bv{A}\bv{y}$. In that setting, the bound is optimal (see \Cref{sec:weakerModels}).

\smallskip\noindent \textbf{Upper Bound.} Our main result shows that $O(\log|\mathcal{F}|)$ is far from optimal for matvec queries. By taking advantage of the multidimensional output of such queries, we give a nearly quadratic improvement in query complexity  \emph{for any finite $\mathcal{F}$}. In particular, in \Cref{sec:finite} we prove:
\begin{theorem}[Finite Family Approximation]\label{thm:robust}
    There is an algorithm (\Cref{alg:binary_search}) that, for any finite matrix family $\mathcal{F}\subset \R^{n\times n}$ with size $|\mathcal{F}|$ and any $\epsilon \in (0,1)$, uses $\tilde{O}(\sqrt{\logfamily}/\epsilon^2)$ matrix-vector product queries\footnote{We use $\tilde{O}(\cdot)$ to suppress polylog factors in the input argument and the failure probability of the algorithm.} 
    with $\bv{A}$ and $\bv{A}^\transpose$ and finds, with high probability, an approximation $\bv{\tilde B} \in \family$ satisfying:
    \begin{align*}
        \|\bv{A} - \bv{\tilde B} \|_{\Fro} < (3+\epsilon)\cdot \min_{\bv{B} \in \family} \|\bv{A}- \bv{B}\|_{\Fro}.
    \end{align*}
\end{theorem}
The algorithm achieving \Cref{thm:robust} strongly relies on the ability to query both $\bv{A}$ and $\bv{A}^\transpose$. It is not hard to prove a lower bound of $\Omega(\log|\mathcal{\F}|)$ if we can only query $\bv{A}$ on one side (see \Cref{sec:weakerModels}). In addition to querying on both sides, our algorithm selects queries adaptively -- the best known non-adaptive method achieves complexity $O(\log |\F|)$. However, we have no lower bounds showing that adaptivity is necessary.

Roughly speaking, we obtain \Cref{thm:robust} by efficiently simulating an iterative sketching method that solves \Cref{prob:intro} using ${O}(\log |\F|)$ \emph{one-sided} queries. Described in \Cref{sec:finite}, the idea is to multiply $\bv{A}$ on the right with a random matrix, $\bv{\Pi}$, with $\tilde{O}(1)$ columns at each iteration, and to eliminate any matrix $\bv{B} \in \mathcal{F}$ for which $\|\bv{A}\bv{\Pi}-\bv{B}\bv{\Pi}\|_{\Fro}$ is large. Using standard results on how well the sketch $\|\bv{A}\bv{\Pi}-\bv{B}\bv{\Pi}\|_\Fro$ approximates $\|\bv{A}-\bv{B}\|_\Fro$, it is easy to show that this approach eliminates all matrices that are not good approximations to $\bv{A}$ after $\tilde{O}(\log |\mathcal{F}|)$ rounds. 

Our main idea is to use an additional left-sketch, $\bv{A}^\transpose \bv{\Psi}$, to identify what fraction of bad approximations $\bv B \in \F$ a particular set of right queries, $\bv{\Pi}$, will eliminate. Importantly, we can do so without ever computing $\bv{A}\bv{\Pi}$ directly -- we only use information in the sketch  $\bv{\Psi}^\T \bv{A}\bv{\Pi}$, which can be computed using our left sketch alone, and no additional matvec queries on the right. We show a dichotomy: either we find that $\bv{\Pi}$ eliminates a very large fraction of bad approximations, in which case we issue the queries and continue. Otherwise, if $\bv{\Pi}$ would remove too few bad approximations (i.e., it provides less information about $\bv{A}$), we prove that it is possible to simulate the result of issuing $\bv{\Pi}$ using the sketch  $\bv{\Psi}^\T \bv{A}\bv{\Pi}$. Balancing these two cases gives our improvement from $O(\log |\F|)$ to $\tilde O(\sqrt{\log |\F|})$.

\smallskip\noindent \textbf{Lower Bound.} 
A natural question is if the bound from \Cref{thm:robust} can be improved. For one, our current techniques only achieve a constant factor approximation; it would be interesting to push the result to achieve approximation factor  $1+\epsilon$. We discuss this and other open problems in \Cref{sec:open_problems}.
Nevertheless,  
 we show that significantly improving the dependence on $|\mathcal{F}|$ in \Cref{thm:robust} is not possible in general for finite matrix families. In \Cref{sec:lower_bounds} we prove:
 
\begin{theorem}[Finite Family Lower Bound]
    \label{thm:two_sided_lower_bound} 
    Any algorithm that solves \Cref{prob:intro} with constant probability for general finite matrix families $\mathcal{F}$ and approximation factor $\gamma$ requires $\Omega(\sqrt{\log |F|/\log\gamma})$ matvec queries.
\end{theorem}

Setting $\gamma = O(1)$, \Cref{thm:two_sided_lower_bound} shows that the dependence on $\log|\F|$ in \Cref{thm:robust} cannot be improved except by log-log factors. 
The ``hard'' family used to prove this result is simple: let $\F$ be a discrete net over the set of all $n \times n$ matrices with bounded norm. Not surprisingly, we require $\Omega(n) = \Omega(\sqrt{\log |\F|})$ matvec queries to learn any meaningful approximation from this family: essentially, we need to read all of $\bv{A}$. 
Our lower bound can also be proven using the more realistic hard instance of constant-rank butterfly matrices, which are commonly used to approximate non-uniform Fourier transforms and other operators \cite{LiuXingGuo:2021}. Such matrices require $\tilde{O}(n)$ parameters to represent, but our best-known butterfly approximation algorithms require $\tilde{O}(\sqrt{n})$ matvecs queries \cite{LiYang:2017}. As discussed further in \Cref{sec:lower_bounds}, our lower bound implies that this is optimal. 

\smallskip\noindent \textbf{Beyond Finite Families.} 
In addition to being a natural starting point for work on general matrix families, one reason to study finite families is that upper bounds immediately imply bounds for infinite families that can be appropriately discretized (i.e., are well-approximated by a finite family). Using standard arguments, we show that \Cref{thm:robust} immediately implies that a near optimal approximation from any family with covering number bounded by $\Gamma$ can be computed with $\tilde{O}(\sqrt{\log \Gamma})$ queries. We conclude new results for well-studied infinite families. For example, there has been interest in approximation via \emph{linearly parameterized families} with dimension $q$, meaning that any matrix in the family can be written as $c_1 \bv{P}_1 + \ldots + c_1 \bv{P}_q$, where $\bv{P}_1, \ldots, \bv{P}_q \in \R^{n \times n}$ are a fixed set of base matrices and $c_1, \ldots, c_q \in \R$ are arbitrary scalar coefficients \cite{EldarLiMusco:2020,HalikiasTownsend:2023,Otto:2023}. Linearly parameterized families include sparse families like banded and block diagonal matrices, displacement-structured matrices like Toeplitz, circulant, and Hankel matrices, and more. In \Cref{sec:covering} we prove:

 \begin{corollary}[Linear Family Approximation]\label{cor:linear}
            For any linearly parameterized matrix family, $\mathcal{L} \subseteq \mathbb{R}^{n \times n}$, with dimension $q$ and any $\alpha,\epsilon \in (0,1)$, it is possible to learn $\bv{\tilde{B}} \in \mathcal{F}$ satisfying, with high probability,
            \begin{align*}
                \|\bv{A} - \bv{\tilde{B}}\|_\Fro \leq (3+\epsilon)\min_{\bv{B}\in \mathcal{L}}\|\bv{A} - \bv{B}\|_\Fro + \alpha\|\bv{A}\|_\Fro,
            \end{align*}
             using $\tilde O \left ({\sqrt{q \log(1/\alpha)}}/{\epsilon^2} \right )$ matrix-vector product queries with $\bv A$ and $\bv A^\T$.
\end{corollary}
Observe that, due to our need to fix a scale for covering $\mathcal{L}$, \Cref{cor:linear} is not a pure multiplicative error result: we incur additive error $\alpha\|\bv{A}\|_\Fro$. While the query complexity depends only logarithmically on $\alpha$, it is still interesting to ask if a purely multiplicative bound can be achieved. In both the one-sided matvec and vector-matrix-vector models, $(1+\epsilon)$ multiplicative error can be obtained with $O(q/\epsilon)$ queries -- i.e., with a quadratically worse dependence on $q$. In particular, leveraging subspace embedding results for Kronecker structured matrices \cite{AhleKapralovKnudsenPaghVelingkerWoodruff2020}, it can be shown that it suffices to return $\min_{\bv{B}\in \mathcal{L}}\|\bv{A}\bv{G} - \bv{B}\bv{G}\|_\Fro$, where $\bv{G}$ is a matrix with i.i.d. sub-Gaussian entries and $O(q/\epsilon)$ columns. Alternatively, results on leverage score sampling and related active regression methods can achieve the same bound with $O(q/\epsilon)$ samples of entries from $\bv{A}$ \cite{ChenPrice:2019a}, which can be obtained in the weaker vector-matrix-vector model. 
Using similar techniques to those used in proving the lower bound in \Cref{thm:two_sided_lower_bound}, it is not hard to argue that the $\sqrt{q}$ dependence offered by \Cref{cor:linear} is unimprovable.

\subsection{Additional Related Work}
\label{sec:related_work}
While we believe our work is the first to address \Cref{prob:intro} for \emph{general matrix families}, as discussed, the problem of learning a matrix from matvec queries has been studied extensively for \emph{specific} structured families, $\F$. We summarize this work below, along with additional related work in slightly different models of computation. 

\smallskip \noindent \textbf{Low-Rank Approximation.}
One of the most familiar settings is when $\family$ is the class of rank-$k$ matrices for some parameter $k$. The query complexity of rank-$k$ approximation has been studied extensively in numerical analysis and theoretical computer science: work on the randomized SVD and related methods \cite{HalkoMartinssonTropp:2011,Woodruff:2014,MuscoMusco:2015} has culminated in a state-of-the-art upper bound of $O(k/\epsilon^{1/3})$ queries to achieve approximation factor $\gamma = 1+\epsilon$ for solving \Cref{prob:intro} \cite{BakshiClarksonWoodruff:2022,MeyerMuscoMusco:2024}. Our best known lower bound for rank-$k$ approximation is  $\Omega(k + 1/\epsilon^{1/3})$ \cite{BakshiNarayanan:2023}. 

\smallskip \noindent \textbf{Sparse Approximation.} There has also been interest in families with a fixed pattern of $s$ non-zeros per row, for a parameter $s$. Such families including diagonal ($s=1$), banded, and block diagonal matrices \cite{BekasKokiopoulouSaad:2007,TangSaad:2011,BastonNakatsukasa:2022,DharangutteMusco:2023}. Matching upper and lower bounds of $\Theta(s/\epsilon)$ queries are known for achieving approximation factor $\gamma = 1+\epsilon$ in this setting \cite{AmselChenHalikias:2024}. Extensive work has also considered sparse matrices with unknown sparsity patterns \cite{CurtisPowellReid:1974,ColemanCai:1986,ColemanMore:1983,WimalajeewaEldarVarshney:2013,SchaferOwhadi:2024,ParkNakatsukasa:2024}.  Sparse matrix approximation finds applications in a variety of settings. One example is in first-order optimization, where methods for learning diagonal and sparse approximations to a Hessian matrix (for which $\bv{x} \rightarrow \bv{A}\bv{x}$ can often computed efficiently using automatic differentiation \cite{Pearlmutter:1994}) are used to implement quasi-second-order methods or to construct preconditioners \cite{BordesBottouGallinari:2009,DauphinVriesBengio:2015,YaoGholamiKeutzer:2020}. 

\smallskip \noindent \textbf{Hierarchical and Block Structured Matrices.}
In computational science, more complex classes of structured matrices are also important. 
For example, among other applications \cite{AmbartsumyanBoukaramBui-Thanh:2020}, methods for approximation via \emph{hierarchical low-rank matrices} (HODLR, HSS, etc.) can be used to construct fast direct solvers for integral equations given access to an implicit solver, like a fast multipole method. 
Algorithms for learning such matrices from a small number of matvec queries, as well as related block low-rank approximations, have been studied heavily in recent years \cite{LinLuYing:2011,Martinsson:2016,LevittMartinsson:2024,LevittMartinsson:2024a,PearceYesypenkoLevitt:2025}. 
Such matrices are also studied in scientific machine learning (SciML), where rank-structured matrix approximation has been utilized to learn the solution operators of linear PDEs given input/output pairs \cite{BoulleTownsend:2023,BoulleHalikiasTownsend:2023,WangTownsend:2023,BoulleTownsend:2024,BoulleHalikiasOtto:2024}. Relatedly, there has been work on learning butterfly matrices from matvecs, which can be characterized by a so-called ``complementary low-rank'' property \cite{LiYangMartin:2015,LiYang:2017,LiuXingGuo:2021}.

Two recent papers established the first multiplicative error bounds in the agnostic setting of \Cref{prob:intro} for HODLR and HSS matrices, two important subclasses of hierarchical low-rank matrices \cite{ChenHalikiasKeles:2025,AmselChenHalikias:2025}.

\smallskip \noindent \textbf{Different Models.}
Beyond matvec query algorithms, there has been extensive work on learning structured matrices (typically, sparse or low-rank) from \emph{linear measurements} of the form $\tr(\bv{A}\bv{M})$ for a measurement matrix $\bv{M}$ \cite{WatersSankaranarayananBaraniuk:2011,CandesPlan:2011,KachamWoodruff:2023}. Most related to our work is the setting when $\bv{M}$ is restricted to be rank-$1$, i.e., $\bv{M} = \bv{y}\bv{x}^\T$ for vectors $\bv{y}$ and $\bv{x}$. In this case, which is referred to as rank-1 matrix sensing \cite{ZhongJainDhillon:2015}, $\tr(\bv{A}\bv{M}) = \bv{x}^\T\bv{A}\bv{y}$, so the problem is equivalent to learning from vector-matrix-vector queries, which have been studied more broadly as a model for linear algebraic computation \cite{RashtchianWoodruffZhu:2020,WimmerWuZhang:2014}. 
A model that sits in between vector-matrix-vector and matrix-vector queries is the \emph{one-sided} matvec model, where we can only access $\bv{A}$ via multiplication on the right (i.e., queries of the form $\bv{x}\rightarrow \bv{Ax}$, but not $\bv{x}\rightarrow \bv{A}^\transpose\bv{x}$). This model arises in some settings, including some operator learning problems in SciML, where it may not be possible to efficiently compute $\bv{A}^\transpose \bv{x}$ \cite{BoulleHalikiasOtto:2024}.

While stronger than the vector-matrix-vector model, we argue in \Cref{sec:weakerModels} that $\Omega(\log |\F|)$ one-sided matvec queries are still needed to learn a generic finite family $\mathcal{F}$ (i.e., quadratically worse than what is possible via \Cref{thm:robust} with two-sided queries). This gap is intuitive: let $\F$ be a finite net over the set of matrices whose first row is an arbitrary unit vector, and whose remaining rows are all zero. This family has size $2^{O(n)}$, but we require $n$ one-sided queries to learn a meaningful approximation -- we essentially need to read each entry in $\bv{A}$'s first row. If two-sided queries are allowed, we could multiply the first standard basis vector by $\bv{A}^\transpose$ to learn the first row with a single query.

Finally, regardless of the specific query model studied, many of the works above focus on the \emph{non-agnostic} or \emph{exact recovery} setting, where $\bv A$ is assumed to be contained in our family $\mathcal{F}$, i.e. $\min_{\bv{B}\in \mathcal{F}}\|\bv{A} - \bv{B}\|_\Fro = 0$. The goal is to exactly identify $\bv A$ -- see e.g., \cite{HalikiasTownsend:2023}. For many matrix classes, the easier exact recovery problem has been tackled first, preceding results in the agnostic setting. In a preliminary version of this paper \cite{AmselAviChenKelesHegdeMuscoMuscoPersson:2025}, we prove exact recovery results for arbitrary finite families using an algorithm that is similar, but simpler, than the one used to prove \Cref{thm:robust}.

\subsection{Notation}
Bold lowercase letters denote vectors and bold uppercase letters denote matrices. 
    For $\bv a \in \R^n$, let $a_i$ denote the $i^\text{th}$ entry. Let $\norm{\bv a}_2$ denote its Euclidean norm. For $\bv B \in \R^{n \times m}$, let $B_{ij}$ denote the entry in the $i^\text{th}$ row and $j^\text{th}$ column. We let $\|\bv{B}\|_\Fro=(\sum_{i=1}^n \sum_{j=1}^m B_{ij}^2)^{1/2}$ denote the Frobenius norm and $\| \bv B \|_2 =  \max_{\bv x: \norm{\bv x}_2 = 1} \norm{\bv B \bv x}_2$ denote the spectral norm. 
    For a random event $\mathcal{A}$ we let $\ones[\mathcal{A}]$ denote the indicator random variable that is $1$ if $\mathcal{A}$ occurs and $0$ otherwise. We let $\mathcal{N}(0,1)$ denote the standard Gaussian distribution and $\mathcal{N}(\bv \mu, \bv C)$ denote the multivariate Gaussian distribution with mean $\bv \mu \in \R^{n}$ and covariance matrix $\bv C \in  \R^{n\times n}$.

\section{Finite Family Matrix Approximation}\label{sec:finite}
    In this section, we establish our main result: an $\tilde{O}(\sqrt{\log |\mathcal{F}|})$ query algorithm for computing a near-optimal approximation to $\bv{A}$ (i.e, for solving \Cref{prob:intro}) from any finite family $\mathcal{F}$. Our proof is based on an efficient ``simulation'' of a natural one-sided query algorithm that only interacts with $\bv{A}$ via right matrix-vector products of the form $\bv x \rightarrow \bv{A}\bv{x}$, and which requires ${O}({\log |\mathcal{F}|})$ queries. Indeed, this is the best possible for a one-sided algorithm (see \Cref{sec:weakerModels}).

     We begin by first introducing and analyzing this one-sided algorithm in \Cref{sec:one_sided}. We then show how to simulate the method with quadratically fewer two-sided queries in \Cref{sec:two_sided}. In both of these sections, we assume we have knowledge of an upper bound, $M$, on the {cost} of the optimum solution, $\OPT = \min_{\bv{B}\in \mathcal{F}} \|\bv{A} - \bv{B}\|_\Fro$. Our error bounds will initially have a dependence on $M$, which we will later remove in \Cref{sec:est}, by showing how to efficiently identify a value of $M$ with $\OPT \leq M \leq (1+\epsilon)\OPT$ for any $\epsilon > 0$. 

    \subsection{One-Sided Algorithms}
    \label{sec:one_sided}
The one-sided method that we simulate differs a bit from the most obvious approach, which we describe first.

    \smallskip \noindent\textbf{The standard approach.}
    An optimal one-sided algorithm for solving \Cref{prob:intro} for finite matrix families can be obtained immediately from prior work on Frobenius norm sketching, or equivalently, stochastic trace estimation (i.e., the Hutchinson-Girard trace estimator) \cite{Hutchinson:1990,Girard:1987,ClarksonWoodruff:2009,AvronToledo:2011,rudelson2013hanson,Roosta-KhorasaniAscher:2015,CortinovisKressner:2022}. In particular, a standard result, following, e.g., from the Hanson-Wright inequality \cite{rudelson2013hanson} is the following:
    \begin{fact}[Frobenius-Norm Sketching]\label{fact:hutch_standard} For any $\epsilon,\delta \in (0,1)$ and  $\ell \geq c\log(1/\delta)/\epsilon^2$  for a fixed constant $c$, let $\bv{\Pi} \in \{-\frac{1}{\sqrt{\ell}},\frac{1}{\sqrt{\ell}}\}^{n\times \ell}$ have i.i.d. scaled Rademacher entries. Then for any matrix $\bv{C}$ with $n$ columns, with probability $> 1-\delta$, 
    \begin{align}
        \label{eq:hutch_standard}
    (1-\epsilon)\|\bv{C}\|_\Fro \leq \|\bv{C}\bv{\Pi}\|_\Fro \leq (1+\epsilon)\|\bv{C}\|_\Fro.
    \end{align}
    \end{fact} 
    To obtain a one-sided solution to \Cref{prob:intro}, we simply use \Cref{fact:hutch_standard} to estimate the error of \emph{every} possible $\bv{B} \in \F$ by setting $\bv{C} = \bv{A} - \bv{B}$. In particular, we can return:
    \begin{align}
        \label{eq:hutch_est}
    \bv{\tilde B} \in \argmin_{\bv{B}\in \F}\|\bv{A}\bv{\Pi}- \bv{B}\bv{\Pi}\|_\Fro. 
    \end{align}
    Computing \eqref{eq:hutch_est} only requires $\ell$ right matrix-vector products of the form $\bv x \rightarrow \bv{Ax}$. Moreover, if we set $\delta' =\delta/|\mathcal{F}|$, and use $\ell = O(\log(1/\delta')/\epsilon^2) = O(\log(|\F|/\delta)/\epsilon^2)$ matvecs, by a union bound, we have that, with probability at least $1-\delta$, $\|\bv{A}\bv{\Pi}- \bv{B}\bv{\Pi}\|_\Fro$ is a $(1\pm\epsilon/3)$-approximation to $\|\bv{A} - \bv{B}\|_\Fro$ \emph{for all} $\bv{B} \in \F$. It follows that $\|\bv{A} - \bv{\tilde B}\|_\Fro \leq (1+\epsilon)\OPT$ for a constant $c$ with high probability. 
    We conclude:
    \begin{claim}
    \label{clm:one_sided_hutch} 
    For any finite matrix family $\mathcal{F}$ and $\epsilon,\delta \in (0,1)$, there is an algorithm that solves \Cref{prob:intro} with approximation factor $\gamma = 1+\epsilon$ and success probability $> 1-\delta$ using $O(\log (|\F|/\delta)/\epsilon^2)$ one-sided matvec queries of the form $\bv x \rightarrow \bv{A}\bv x$. 
    \end{claim}

    As discussed, it is possible to achieve the same bound as in \Cref{clm:one_sided_hutch} using just vector-matrix-vector queries of the form $\bv x, \bv y \rightarrow \bv x^\T \bv{A} \bv y$. 
    We prove this in \Cref{sec:weakerModels}, where we also show that the bound is essentially tight, for both one-sided matvec queries and vector-matrix-vector queries.

    We briefly note that one might hope to improve the dependence on $\epsilon$ in \Cref{clm:one_sided_hutch} using improvements to the Hutchinson-Girard estimator like Hutch++ and related techniques \cite{MeyerMuscoMusco:2021,PerssonCortinovisKressner:2022,EpperlyTroppWebber:2024,JiangPhamWoodruff:2024}. However, all of these methods require matrix-vector multiplications with $(\bv{A} - \bv{B})^\T(\bv{A} - \bv{B})$. It is unclear how to compute this in a query-efficient manner, since we would need to compute $\bv{A}^\T\bv{B}\bv{x}$ for all $\bv{B} \in \F$.

    \smallskip \noindent\textbf{An alternative iterative approach.} While simple, it is not clear how to simulate the above approach with a smaller number of two-sided queries. Instead, we consider a closely related alternative method based on \emph{iteratively refining} $\F$. At every step of the algorithm, we use a very small sketch $\bv{\Pi}$ to remove a large fraction of candidate approximations from $\mathcal{F}$ until, after roughly $O(\log|\F|)$ rounds, we converge on a near-optimal approximation.

    This approach is formalized as \Cref{alg:one_side_refinement}. The method performs $\log |\F|/\log \log |\F|$ iterations of refinement, using a sketch of size $O(\log \log |\F|)$ in each iteration, for a total query complexity of $O(\log|\F|)$, matching the standard approach of \Cref{clm:one_sided_hutch}.

    \begin{algorithm}[H]
        \caption{One-Sided Iterative Candidate Refinement}\label{alg:one_side_refinement}    
        \begin{algorithmic}[1]
            \Require Finite matrix family $\mathcal{F} \subset \mathbb{R}^{n \times n}$, failure probability $\delta \in (0,1)$, accuracy parameter $\epsilon \in (0,1)$, upper bound $M \geq \min_{\bv{B}\in \F} \|\bv{A} - \bv{B}\|_\Fro = \OPT$.
            \Ensure Matrix $\bv{\tilde B}\in \F$ satisfying $\|\bv{A} - \bv{\tilde B}\|_\Fro \leq (1 + \epsilon)\cdot M$ with probability at least $1-\delta$.
            \algrule
            \State Initialize candidate set $\mathcal{C}_0=\mathcal{F}$.
Set $\ell = \frac{c}{\epsilon^2}\log(\log|\F|/\delta)$ for sufficiently large constant $c$. Set $T = \frac{\log |\F|}{\log \log |\F|}$.\label{line:alg1Init}
            \For{$i = 1,\ldots, T$}    
                \State Draw $\bv{\Pi} \in \{-1/\sqrt{\ell},1/\sqrt{\ell}\}^{n\times \ell}$ with i.i.d. scaled Rademacher entries.
                \State Compute $\bv{Z} = \bv{A}\bv{\Pi}$. \Comment{Requires $\ell$ matvec queries.}
                \State Let $\mathcal{C}_i = \left \{\bv{B} \in \mathcal{C}_{i-1} : \|\bv{Z} - \bv{B}\bv{\Pi}\|_\Fro \leq (1+\epsilon/2)\cdot M \right \}$.\label{line:alg1Check} 
            \EndFor
            \State \Return Any remaining $\bv{B} \in \mathcal{C}_{T}$.\label{line:alg1Return} 
        \end{algorithmic}
    \end{algorithm}    

    Note that \Cref{alg:one_side_refinement}
    assumes knowledge of an upper bound $M$ on the cost of the optimal solution $\min_{\bv{B}\in \F} \|\bv{A} - \bv{B}\|_\Fro = \OPT$. Using a standard binary search method, we will show how to efficiently compute an accurate upper bound in \Cref{sec:est}, which leads to an algorithm with a relative error guarantee. As written, \Cref{alg:one_side_refinement} guarantees error at most $(1+\epsilon)M \geq (1+\epsilon)\OPT$.

    The proof of correctness is straightforward. We make two observations. First, consider any fixed $\bv{B}^* \in \argmin_{\bv{B}\in \F} \|\bv{A} - \bv{B}\|_\Fro$. Via the standard Frobenius norm sketching bound of \Cref{fact:hutch_standard}, as long as the constant $c$ is set large enough in \Cref{line:alg1Init}, at each iteration of the loop, we have that the estimate $\|\bv{Z} - \bv{B}^*\bv{\Pi}\|_\Fro = \|\bv{A}\bv{\Pi} -\bv{B}^*\bv{\Pi}\|_\Fro$ is a $(1\pm \epsilon/2)$ multiplicative approximation to $\|\bv{A} - \bv{B}^*\|_\Fro$ with probability at least $1 - \delta / \log |\F|$. By a union bound, this holds at every iteration with probability at least $1-\delta$. So,  with high probability, $\bv{B}^*$ survives the check at \Cref{line:alg1Check} in all iterations, and $\bv{B}^* \in \mathcal{C}_{T}$. In other words, the optimal solution remains as a candidate for output in \Cref{line:alg1Return}.
    
    Therefore, we obtain an accurate solution $\bv{\tilde B}$ with $\|\bv{A} -\bv{\tilde B}\|_\Fro \leq (1+\epsilon) \cdot M$ as long as we can ensure that $\mathcal{C}_{T}$ 
    contains no ``bad'' matrices $\bv{B}$ for which $\|\bv{A} - \bv{B}\|_\Fro > (1+\epsilon) \cdot M$. To argue that this is indeed the case, let $\mathcal{B}\subset \mathcal{F}$ denote this set of ``bad'' matrices with $\|\bv{A} - \bv{B}\|_\Fro > (1+\epsilon) \cdot M$. 
    Fix a particular ${\bv{B}} \in \mathcal{B}$. Again invoking \Cref{fact:hutch_standard}, as long as the constant $c$ is chosen sufficiently large in \Cref{line:alg1Init}, at each iteration $i$, $\|\bv{Z} - {\bv{B}}\bv{\Pi}\|_\Fro > (1-\epsilon/6)\|\bv{A} - {\bv{B}}\|_\Fro > (1-\epsilon/6)\cdot (1+\epsilon) \cdot M> (1+\epsilon/2) \cdot M$  
    with probability at least $1-{\delta}/{\log|\F|}$. I.e., $\bv B$ is eliminated at \Cref{line:alg1Check} with probability at least $1-{\delta}/{\log|\F|}$ in each iteration. Therefore, after $T  = \frac{\log|\F|}{\log \log |\F|}$ iterations, the probability that ${\bv{B}}$ remains in the candidate set is at most $({\delta}/{\log|\F|} )^{\frac{\log|\F|}{\log \log |\F|}} < \frac{\delta}{|\mathcal{F}|}$. 
    Via a union bound, since $|\mathcal{B}| < |\mathcal{F}|$, no ${\bv{B}} \in \mathcal{B}$ remains in $\mathcal{C}_{T}$ with probability at least $1-\delta$, and thus we always return $\bv{\tilde B} \notin \mathcal{B}$.
    After noting that \Cref{alg:one_side_refinement} uses $\ell \cdot T = O \left (\frac{\log(\log|\F|/\delta)}{\epsilon^2} \cdot \frac{\log|\F|}{\log\log|\F|} \right ) = O \left (\frac{\log|\F| \cdot \log(1/\delta)}{\epsilon^2}\right)$ queries,
    we conclude:
    \begin{claim}
    \label{claim:one_sided_cutting_method} 
    For any finite matrix family $\mathcal{F}$, \Cref{alg:one_side_refinement} with inputs $\epsilon,\delta \in (0,1)$ and $M \ge \min_{\bv{B}\in \F} \|\bv{A} - \bv{B}\|_\Fro = \OPT$
     solves \Cref{prob:intro} with error $\gamma = (1+\epsilon) \cdot \frac{M}{\OPT}$ and success probability $> 1-\delta$ using ${O}(\log|\F| \cdot \log(1/\delta)/\epsilon^2)$ one-sided matvec queries of the form $\bv x \rightarrow \bv{A}\bv x$.
    \end{claim}
As mentioned, in \Cref{sec:est} we will show how to compute $M$ with $M/\OPT \le 1+\epsilon$, giving a $1+\epsilon$ relative error guarantee for \Cref{alg:one_side_refinement} after plugging into \Cref{claim:one_sided_cutting_method}.

    \subsection{Near Quadratic Improvement via Two-Sided Simulation}
    \label{sec:two_sided}
    In this section, we show how to use matvecs with both $\bv{A}$ and $\bv{A}^\T$ to solve \Cref{prob:intro} with just $\tilde{O}(\sqrt{\log|\mathcal{F}|})$ queries. Our approach is based on an efficient ``simulation'' of \Cref{alg:one_side_refinement} from the previous section.
    We first sketch the main ideas behind this simulation in \Cref{sec:sketch}.
    We then formally present the approach as \Cref{alg:two_sided_simulation} and analyze it in \Cref{sec:formal}, yielding \Cref{thm:main_opt_known}. Combining \Cref{thm:main_opt_known} with an efficient method for estimating an accurate upper bound $M \geq \OPT$, given in \Cref{sec:est},  yields our main result, \Cref{thm:robust}.

    \subsubsection{Proof Overview}\label{sec:sketch}
For simplicity of asymptotic notation,  we assume throughout this proof sketch that both our error parameter $\epsilon$ and the failure probability $\delta$ are fixed constants. We also suppress all $\log \log |\F|$ factors with $\tilde O(\cdot)$ notation, and write $\log |\F|/\log\log|\F| = O(\log |\F|)$.

    Recall that the goal of  \Cref{alg:one_side_refinement} is to continually shrink the candidate set $\mathcal{C}$ until it no longer contains any ``bad'' matrices that are not good approximations to $\bv{A}$. In the proof of \Cref{claim:one_sided_cutting_method}, we argue that for every block of $\ell = {O}(\log\log |\F|)$ right queries, we reduce the set of bad candidates by an $O(1/\log|\F|)$ factor, eliminating all bad candidates after $\log |\F|/\log \log |\F|$ iterations. This leads to an overall query complexity of  ${O}(\log |\F|)$. 
    
    To obtain an algorithm that uses only $\tilde{O}(\sqrt{\log|\mathcal{F}|})$ queries, we need to cut down the set of candidates far more aggressively, ideally by an $O(1/2^{\sqrt{\log|\F|}})$ factor with each block of queries.
    To do so, our simulation 
   leverages the fact that each iteration of \Cref{alg:one_side_refinement} falls under one of two cases, depending on the random sketching matrix $\bv \Pi$:

    \medskip
    \noindent\textbf{Case 1.}
    First, if we get lucky, $\bv \Pi$ might be a highly ``productive'' sketch. In particular, 
    recalling that are targeting constant error in this proof sketch, we could have that, for a constant $c$,
    $\|\bv{A}\bv{\Pi} - \bv{B}\bv{\Pi}\|_\Fro > c M$ for all but an $O(1/2^{\sqrt{\log|\F|}})$ fraction of $\bv{B}\in \mathcal{C}_i$. In this case, we can simply issue the right queries $\bv{\Pi}$ and remove a large fraction of remaining candidates.

    \medskip
    \noindent\textbf{Case 2.}
    Alternatively, if we are not so lucky in our choice of $\bv{\Pi}$, there could be \emph{many  candidates} $\bv{B}$ for which $\|\bv{A}\bv{\Pi}-\bv{B}\bv{\Pi}\|_\Fro \le c M$ (at least an $\Omega(1/2^{\sqrt{\log|\F|}})$ fraction).
    However, we can take advantage of this situation to filter the candidate set $\mathcal C_i$ \emph{with no right queries at all}.
    In particular, suppose hypothetically that we knew that a certain matrix $\bv{Y}$ had  $\|\bv Y-\bv{A}\bv{\Pi}\|_\Fro \le c  M$, we could use it as a proxy, discarding candidates $\bv{B}$ for which $\|\bv{Y} - \bv{B}\bv{\Pi}\|_\Fro \ge 2c M$, without needing to compute $\bv A \bv \Pi$ and thus make any right queries to $\bv A$.
    By the triangle inequality, the discarded matrices would all have $\|\bv{A \Pi} - \bv{B}\bv{\Pi}\|_\Fro \ge c M$, and we would still successfully eliminate a large fraction of bad candidates with this proxy rule.

    Our key insight is that, if we find ourselves in Case 2, it is easy to construct such a $\bv{Y}$, since there are many candidates $\bv{B}$ for which $\bv{Y} = \bv{B}\bv{\Pi}$ is close to $\bv{A}\bv{\Pi}$.
    In particular, we just need to find a suitable ``representative'', $\bv{R}\in \mathcal{C}_i$ such that $\|\bv{A}\bv{\Pi} - \bv{R}\bv{\Pi}\|_\Fro \le c  M$. 
    To do so, we uniformly sample roughly $q = O(2^{\sqrt{\log|\F|}})$ candidates $\bv{B}_1, \ldots, \bv{B}_q$ from $\mathcal C_i$. By the assumption of Case 2, at least an $\Omega (1/q)$ fraction of the candidates in $\mathcal C_i$ have $\|\bv{A}\bv{\Pi} - \bv{B}\bv{\Pi}\|_\Fro \leq c M$. So, with high probability, we will have at least one such candidate in our sample. 
    We can use Frobenius norm sketching (\Cref{fact:hutch_standard}) with \emph{left} queries to identify this representative candidate. That is, we sample a left sketching matrix $\bv{\Psi}$ with $\tilde O(\log q) = \tilde{O}(\sqrt{\log |\F|})$ columns and check the value of $\|\bv{\Psi}^\T\bv{A}\bv{\Pi} - \bv{\Psi}^\T\bv{B}_j\bv{\Pi}\|_\Fro$ for each $j = 1, \ldots, q$. By a union bound over our $q$ samples, this estimate will be a constant factor approximation to $\|\bv{A}\bv{\Pi} - \bv{B}_j\bv{\Pi}\|_\Fro$ with high probability for all $\bv B_j$, and thus enough to identify our representative $\bv R \in \{\bv B_1,\ldots,\bv B_q\}$.
          
    Importantly, the above process requires \emph{no right queries to $\bv A$}, since we can compute $\bv{\Psi}^\T \bv A$ using left queries, and then just multiply this matrix by $\bv \Pi$ to compute  $\|\bv{\Psi}^\T\bv{A}\bv{\Pi} - \bv{\Psi}^\T\bv{B}_j\bv{\Pi}\|_\Fro$.
    With our representative $\bv{R}$ in hand, we can filter out candidates where $\|\bv{R\Pi}-\bv{B\Pi}\|_\Fro$ is large, eliminating a large fraction of bad candidates using this proxy rule (i.e., all but a $O(1/\log|\F|)$ fraction as in \Cref{alg:one_side_refinement}) with high probability.

    \paragraph{Putting the cases  together.}
    The above case analysis suggests the following algorithm:
    We first precompute the left sketch $\bv{\Psi}^\T \bv{A}$ using a total of $\tilde O(\sqrt{\log |\mathcal F|})$ left queries. This sketch will be used to search for a representative $\bv R$ in each of the $O(\log|\F|)$ iterations of \Cref{alg:one_side_refinement}.
    At each iteration, we draw a $\bv{\Pi}$ and determine if it falls under Case 1 or Case 2. This determination is made by searching for a representative $\bv{R}$. If none is found, we conclude we are in Case 1.
    In Case 1, we issue $\ell = \tilde O(1)$ right queries and eliminate all but an $O(1/2^{\sqrt{\log |\mathcal F|}})$ fraction of the remaining candidates. We can encounter this case at most $O(\sqrt{\log |\F|})$ times before removing all candidates, and thus will issue at most $\tilde O(\sqrt{\log |\F|})$ right queries in total.
    In Case 2, we issue no new queries but still eliminate all but a $O(1/\log|\F|)$ fraction of bad candidates, meaning that we still terminate after at most $O(\log |\F|)$ iterations. 
    Regardless, we issue at most $\tilde O(\sqrt{\log |\mathcal F|})$ right and left queries in total, and identify a constant factor solution to \Cref{prob:intro} with high probability.

\paragraph{Independence Across Iterations.} The description above captures the key idea behind our $\tilde O(\sqrt{\log |\F|})$ query algorithm. However, it ignores a key issue:
to obtain the claimed complexity, we need to \emph{reuse} the same  left sketching matrix $\bv{\Psi}$ across all iterations of the algorithm. 
However, for iteration $i > 0$, the candidate set $\mathcal{C}_i$ could depend on $\bv{\Psi}$, since it may depend on the choice of representative $\bv{R}$ at an earlier iteration, which is identified using the left sketch $\bv{\Psi}^\T \bv A$. Hence, the set of randomly sampled potential representatives, $\bv{B}_1, \ldots, \bv{B}_q$ may depend on $\bv{\Psi}$ too.
Due to this loss of independence, we cannot directly argue that $\|\bv{\Psi}^\T\bv{A}\bv{\Pi} - \bv{\Psi}^\T\bv{B}_j\bv{\Pi}\|_\Fro$ is a good approximation to $\|\bv{A}\bv{\Pi} - \bv{B}_j\bv{\Pi}\|_\Fro$ for each $\bv B_j\in \{\bv B_1, \ldots, \bv B_q\}$ as required to identify a good representative $\bv R$.

To avoid this issue, we slightly modify the procedure. Instead of drawing a single right sketch for each iteration, we draw $r = O(\log|\F|)$ sketches for each iteration, $\bv{\Pi}^{1}, \ldots, \bv{\Pi}^{r}$.
We then attempt to find representatives $\bv R^{1}, \ldots \bv R^{r}$ (via sampling as described above) for each of these sketches.
If we fail to find a representative for a particular $\bv{\Pi}^{j}$, it means that $\bv{\Pi}^{j}$ lands in Case 1, and we can issue these queries to filter the candidate set down to just an $O(1/2^{\sqrt{\log |\F|}})$ fraction of its size.
Otherwise, if we succeed in finding a representative for \emph{every} sketch $\bv{\Pi}^1,\ldots,\bv{\Pi}^{r}$, we can use these representatives to filter the candidate set for free according to \emph{all of these sketches at once}. 
For $r = O(\log|\F|)$, this process filters out all the remaining bad candidates with high probability in a single shot, and the algorithm terminates.
The benefit of this approach is that in Case 1, the set of candidates that survives the filtering depends on $\bv{\Psi}$ only through the choice of the sketching matrix $\bv{\Pi}^j$ for which we could not find a representative. Since there are just $r = O(\log|\F|)$ sketching matrices to choose from at each iteration, we can control the number of candidate sets that we might possibly want to sketch with $\bv{\Psi}$ and use a union bound to guarantee that the sketch is accurate for all possible such candidate sets. In Case 2, while we use representatives $\bv R^1,\ldots,\bv R^r$ that depend on $\bv \Psi$, we terminate in a single shot, and thus there are no future iterations where dependency issues could arise.

\subsubsection{Formal Analysis}\label{sec:formal}

We next formalize the approach described in \Cref{sec:sketch} as \Cref{alg:two_sided_simulation}, and give a full analysis in \Cref{thm:main_opt_known}. 

\begin{algorithm}[h!]
        \caption{Two-Sided Iterative Candidate Refinement}\label{alg:two_sided_simulation}    
        \begin{algorithmic}[1]
            \Require Finite matrix family $\mathcal{F} \subset \mathbb{R}^{n \times n}$, failure probability $\delta \in (0,1)$, accuracy parameter $\epsilon \in (0,1)$, upper bound $M \geq \min_{\bv{B}\in \F} \|\bv{A} - \bv{B}\|_\Fro = \OPT$.
            \Ensure Matrix $\bv{\tilde B}\in \F$ satisfying $\|\bv{A} - \bv{\tilde B}\|_\Fro \leq (3 + \epsilon)\cdot M$ with probability at least $1-\delta$.
            \algrule
            \State Let $c$ be a sufficiently large universal constant and set:
            \begin{align*}
            &m = \frac{c}{\epsilon^2}\sqrt{\log |\F|}\cdot \log (\log |\F|/\delta) &&{\color{gray} \triangleright\text{ Left queries}} \\
            &\ell = \frac{c}{\epsilon^2}\log(\log|\F|/\delta) &&{\color{gray}\triangleright\text{ Queries per right sketch}} \\
            &r = \log |\F|/\log\log|\F| &&{\color{gray}\triangleright\text{ Right sketches per iteration}} \\
            &q = 2^{1.5 \cdot \sqrt{\log|\F|}\cdot \log(\log|\F|/\delta)} &&{\color{gray}\triangleright\text{ Candidates sampled per iteration}}
            \end{align*}
            \State Draw the following sketching matrices with i.i.d Rademacher entries:
            \begin{align*}
           \bv{\Psi} &\in \{-1/\sqrt m,1/\sqrt m\}^{n\times m} &&{\color{gray}\triangleright\text{ Left sketching}}\\
            \bv{\Pi}^{i,j} &\in \{-1/\sqrt \ell,1/\sqrt \ell\}^{n\times \ell} 
            \quad \text{for each} (i,j) \in \{0, \ldots, \sqrt{\log |\F|}\} \times \{1, \ldots, r\}
            &&{\color{gray}\triangleright\text{ Right sketching}}
            \end{align*}
            \State Compute $\bv{W} = \bv{\Psi}^\T\bv{A}$. 
            \Comment{Requires $m$ matvecs with $\bv{A}^\T$}
            \State Initialize candidate set $\mathcal{C}_0=\mathcal{F}$.
            \For{$i = 0,\ldots, \sqrt{\log |\F|}$}
                \State Sample set of possible representatives $\mathcal{R}_{i}$ uniformly without replacement from $\mathcal{C}_{i}$, with $|\mathcal{R}_{i}| = \min\left(q, |\mathcal{C}_i|\right)$.\label{line:sample}
                \State For each $j\in \{1, \ldots, r\}$, let \begin{align*}\bv{R}^{i,j} &= \argmin_{\bv{B}\in \mathcal{R}_{i}} \left\|\bv{W}\bv{\Pi}^{i,j} - \bv{\Psi}^\T\bv{B}\bv{\Pi}^{i,j}\right\|_\Fro &&{\color{gray}\triangleright\text{ Best representative for sketch $\bv{\Pi}^{i,j}$}}\\
                E_{i,j} &= \left\|\bv{W}\bv{\Pi}^{i,j} - \bv{\Psi}^\T\bv{R}^{i,j}\bv{\Pi}^{i,j}\right\|_\Fro &&{\color{gray}\triangleright\text{ Error of best representative}}
                \end{align*}\label{line:minCheck}
                \If{$E_{i,j} \le (1+\epsilon/6)\cdot M$ for all $j \in \{1,\ldots,r\}$}\label{line:conditional} \Comment{Good representative found for all sketches.}
                       \State \label{line:return} \Return $\displaystyle\bv{\tilde B} \in \argmin_{\bv{B} \in \mathcal{C}_{i}}\max_{j\in \{1, \ldots, r\}} \|\bv{R}^{i,j}\bv{\Pi}^{i,j} - \bv{B}\bv{\Pi}^{i,j}\|_\Fro.$ \Comment{Filter down to best candidate in a single shot}
                 \Else
                    \State \label{step:special j}Select any $j^*_i$ for which $E_{i,j^*_i} > (1+\epsilon/6)\cdot M$. \Comment{$\bv{\Pi}^{i,j^*_i}$ is a highly productive sketch. Use it to filter $\mathcal C_i$.}
                    \State Compute $\bv{Z}_i = \bv{A}\bv{\Pi}^{i,j^*_i}$. \Comment{Requires $\ell$ matvecs with $\bv{A}$.}
                    \State Update the candidate set: \label{line:filter}\begin{equation}\label{eq:candidate_update}\mathcal{C}_{i+1} = \Big\{\bv B \in \mathcal{C}_i \, : \, \|\bv{Z}_i - \bv{B}\bv{\Pi}^{i,j^*_i}\|_\Fro \leq (1+\epsilon/12)\cdot M \Big\}.\end{equation}
                \EndIf
            \EndFor
            \State \Return \texttt{Failed}\label{line:fail} \Comment{With high probability, this is never reached.}
        \end{algorithmic}
    \end{algorithm}

    \begin{theorem}
    \label{thm:main_opt_known}
    For any finite family of matrices $\mathcal{F}$ and target matrix $\bv{A}$, \Cref{alg:two_sided_simulation} run with inputs $\epsilon, \delta \in (0,1)$ and $M \geq \min_{\bv{B}\in \mathcal{F}} \|\bv{A} - \bv{B}\|_\Fro = \OPT$, returns $\bv{\tilde B}\in \F$ satisfying, with probability at least $1-\delta$,
    \begin{align*}
    \|\bv{A} - \bv{\tilde B}\|_\Fro \leq (3 + \epsilon) \cdot M.
    \end{align*}
    Further, the algorithm uses $O(\sqrt{\log|\F|} \cdot \log(\log|\F|/\delta)/\epsilon^2) = \tilde{O}(\sqrt{\log|\F|}/\epsilon^2)$ matvec queries with $\bv{A}$ and $\bv{A}^\T$.
\end{theorem}

\begin{proof}
As discussed, a challenging aspect of our analysis is ensuring that the left sketch $\bv{\Psi}$ can  be reused across different iterations of the algorithm. We begin by arguing that this is the case. 
    
\paragraph{Left Sketch Analysis.}
Our goal is to prove that $\bv \Psi$ accurately estimates the distance between $\bv A \bv{\Pi}^{i,j}$ and $\bv B \bv{\Pi}^{i,j}$ for any candidate representative $\bv{B}$ tested at Line 7 of the algorithm. Concretely, we claim that, with probability $\geq 1-\delta$, the following holds for all $i \in \{1, \ldots, \sqrt{\log |\F|}\}$, $j \in \{1, \ldots, r\}$, and $\bv{B} \in \mathcal{R}_{i}$:
\begin{align}
    \label{eq:left_sketch_guarantee}
    (1-\epsilon/16) \|\bv{A}\bv{\Pi}^{i,j} - \bv{B}\bv{\Pi}^{i,j}\|_\Fro \leq \|\bv{\Psi}^\T\bv{A}\bv{\Pi}^{i,j} - \bv{\Psi}^\T\bv{B}\bv{\Pi}^{i,j}\|_\Fro \leq (1+\epsilon/16) \|\bv{A}\bv{\Pi}^{i,j} - \bv{B}\bv{\Pi}^{i,j}\|_\Fro.
\end{align}

We prove \eqref{eq:left_sketch_guarantee} by combining 
\Cref{fact:hutch_standard} with a union bound over all possible candidates $\bv B$ that may be considered. First, fix the randomness used to generate each $\bv{\Pi}^{i,j}$. 
Define:
\[ \mathcal C[j] = \Big\{\bv B \in \mathcal{C}_0 \, : \, \|\bv{Z}_i - \bv{B}\bv{\Pi}^{0,j}\|_\Fro \leq (1+\epsilon/12)\cdot M \Big\} \]
That is, $\mathcal C[1], \ldots, \mathcal C[r]$ are the possible values that $\mathcal C_1$ may be set to if \Cref{line:filter} is reached in the first iteration of the for loop. $\mathcal C_1 = \mathcal C[j^*_1]$, where $j^*_1$ is selected at \cref{step:special j} and is a random variable depending on $\bv{\Psi}$. 
In general, define
\[
\mathcal C[j_0, \ldots, j_{i-1}, j_{i}] = \Big\{\bv B \in \mathcal C[j_0, \ldots, j_i] \, : \, \|\bv{Z}_{i-1} - \bv{B}\bv{\Pi}^{i-1,j_{i}}\|_\Fro \leq (1+\epsilon/12)\cdot M \Big\}. 
\]
$\mathcal C[j_0, \ldots, j_i]$ is one possible value that $\mathcal C_{i}$ may be set to.
Let $\mathcal{R}[j_0, \ldots, j_{i}]$ be a random sample of $\min(q,|\mathcal C[j_0, \ldots, j_i]|)$ candidates from $\mathcal C[j_0, \ldots, j_i]$ -- i.e., the sample that will be taken at \Cref{line:sample} if $\mathcal{C}_i = C[j_0, \ldots, j_i]$. Note that $\mathcal{R}[j_0, \ldots, j_{i}]$ is independent from $\bv{\Psi}$. Now, let $\mathcal B$ denote the set of all possible candidate sketches $\bv B \bv{\Pi}^{i,j}$ that may be considered at \Cref{line:minCheck} (i.e., all matrices for which we may need \eqref{eq:left_sketch_guarantee} to hold):
\begin{align}\label{eq:masterSet}\mathcal B = \left \{\bv B \bv{\Pi}^{i,j} : \bv B \in {R}[j_0, \ldots, j_{i}] \text{ for some } i \in \{1,\ldots \sqrt{\log|\F|}\}, \text{ and }(j,j_0,\ldots,j_i) \in \{1,\ldots,r\}^{i+1} \right \}.
\end{align}
$\mathcal B$ is a random set depending on the choices of $R[j_0, \ldots, j_{i}]$, but importantly, is independent of $\bv \Psi$. Further the number of sets $R[j_0,\ldots,j_i]$ is bounded by $1 + r + r^2 + \ldots + r^{\sqrt{\log|\F|}}  = O\left(r^{\sqrt{\log|\F|}}\right)$. Since each such set contains at most $q$ elements and since $j$ in \eqref{eq:masterSet} can take $r$ values, $|\mathcal{B}| = O\left(q\cdot r^{\sqrt{\log|\F|}+1}\right)$.

Applying \Cref{fact:hutch_standard} with failure probability $\delta' = \delta/|\mathcal{B}|$, it  follows from a union bound that \eqref{eq:left_sketch_guarantee} holds all $i \in \{1, \ldots, \sqrt{\log |\F|}\}$, $j \in \{1, \ldots, r\}$, and $\bv{B} \in \mathcal{R}_{i}$ with probability at least $1 - \delta$ as long as $\bv{\Psi}$ has
\[
m = \Omega\left(\frac{1}{\epsilon^2}\log\left(q\cdot r^{\sqrt{\log|\F|} + 1}/\delta\right)\right) = \Omega\left(\frac{1}{\epsilon^2}\sqrt{\log |\F|} \cdot \log(\log |\F|/\delta) \right)
\text{ columns.}
\]

        \paragraph{Optimal Solution Survives.}
        Fix some $\bv{B}^* \in \argmin_{\bv{B}\in \mathcal{F}} \|\bv{A} - \bv{B}\|_\Fro$. Our next goal is to prove that, as in the proof of \Cref{claim:one_sided_cutting_method}, $\bv{B}^*$ never gets removed from the candidate set $\mathcal{C}_i$. In particular, with probability at least $1-\delta$, 
        \begin{align}
            \label{eq:bstar_never_removed}
    \bv{B}^* \in \mathcal{C}_i \text{ for all } \mathcal{C}_i\text{ formed by \Cref{alg:two_sided_simulation} before termination.}
        \end{align} 
        To prove \eqref{eq:bstar_never_removed}, note that there are $O(\log^{1.5}|\F|)$ right sketches $\bv{\Pi}^{i,j}$, and each has $\ell = \Omega\left(\log(\log|\F|/\delta)/\epsilon^2\right)$ columns. Applying \Cref{fact:hutch_standard} and a union bound, we have that, with probability $\geq 1-\delta$, for all $i,j$, 
        \begin{align}
            \label{eq:norm_of_bstar}
        \|\bv{A}\bv{\Pi}^{i,j} - \bv{B}^*\bv{\Pi}^{i,j}\|_\Fro \leq (1+\epsilon/12) \cdot M.
        \end{align}
        Accordingly, $\bv{B}^*$ always satisfies the condition \eqref{eq:candidate_update} on \cref{line:filter}, so is always included in $\mathcal{C}_{i}$.

        \paragraph{Accuracy.}
        We now show that whenever the algorithm returns an output on \cref{line:return}, it is an accurate approximation to $\bv{A}$ with high probability. Specifically, we aim to show that $\|\bv{A} - \bv{\tilde B}\|_\Fro < (3+\epsilon)\cdot M$ with probability at least $1-\delta$.

        Suppose the algorithm terminates at iteration $i$.
        This occurs because $\|\bv{W}\bv{\Pi}^{i,j} - \bv{\Psi}^\T\bv{R}^{i,j}\bv{\Pi}^{i,j}\|_\Fro = \|\bv{\Psi}^\T\bv{A}\bv{\Pi}^{i,j} - \bv{\Psi}^\T\bv{R}^{i,j}\bv{\Pi}^{i,j}\|_\Fro \leq (1+\epsilon/6)\cdot M$ for all $j\in \{1, \ldots, r\}$, triggering the conditional at \Cref{line:conditional}.
         Applying \eqref{eq:left_sketch_guarantee}, we conclude that with probability $1-\delta$, the following holds for all $j$:
        \begin{align}
            \label{eq:error_of_r}
         \|\bv{A}\bv{\Pi}^{i,j} - \bv{R}^{i,j}\bv{\Pi}^{i,j}\|_\Fro \leq \frac{1+\epsilon/6}{1-\epsilon/16}\cdot M \leq (1+\epsilon/4)\cdot M.
        \end{align}
        Using the triangle inequality to combine this bound with that of \eqref{eq:norm_of_bstar}, we then have that, for all $j$,
        \begin{align}
        \|\bv{R}^{i,j}\bv{\Pi}^{i,j} - \bv{B}^*\bv{\Pi}^{i,j} \|_\Fro \leq  \|\bv{A}\bv{\Pi}^{i,j} - \bv{R}^{i,j}\bv{\Pi}^{i,j}\|_\Fro +  \|\bv{A}\bv{\Pi}^{i,j} - \bv{B}^*\bv{\Pi}^{i,j}\|_\Fro \leq (2+\epsilon/2)\cdot M.
        \end{align}
        Additionally, by \eqref{eq:bstar_never_removed}, we have that $\bv{B}^* \in \mathcal{C}_i$ with probability at least $1-\delta$.
        Therefore, we return $\bv{\tilde B}$ for which the following holds for all $j$:
        \begin{align}
        \|\bv{R}^{i,j}\bv{\Pi}^{i,j} - \bv{\tilde B}\bv{\Pi}^{i,j} \|_\Fro
        &\le \min_{\bv{B} \in \mathcal{C}_{i}}\max_{j\in \{1, \ldots, r\}} \|\bv{R}^{i,j}\bv{\Pi}^{i,j} - \bv{B}\bv{\Pi}^{i,j}\|_\Fro \nonumber \\
        &\leq \max_{j\in \{1, \ldots, r\}} \|\bv{R}^{i,j}\bv{\Pi}^{i,j} - \bv{B}^*\bv{\Pi}^{i,j}\|_\Fro \nonumber\\
        &\leq (2+\epsilon/2)\cdot M.\label{eq:upper bound on output}
        \end{align}
        Now, we would like to prove that $\|\bv{A} - \bv{\tilde B}\|_\Fro \le (3+\epsilon)\cdot M$ with high probability. It suffices to argue that, with probability at least $1-\delta$, for all ``bad'' $\bv B$ with $\|\bv{A} - \bv{B}\|_\Fro > (3+\epsilon)\cdot M$,
        \begin{align}
            \label{eq:bad_matrices_look_bad}
        \|\bv{R}^{i,j}\bv{\Pi}^{i,j} - \bv{B}\bv{\Pi}^{i,j} \|_\Fro > (2+\epsilon/2)\cdot M \text{ for at least one } j\in \{1, \ldots, r\}.
        \end{align}
        Thus, by \eqref{eq:upper bound on output},  
        $\bv{\tilde B}$ cannot be bad. By triangle inequality and \eqref{eq:error_of_r}, to prove \eqref{eq:bad_matrices_look_bad}, it suffices to show that, for all bad $\bv{B}$, 
        \begin{align}
            \label{eq:bad_matrices_look_bad_after_triangle_inequality}
        \|\bv{A}\bv{\Pi}^{i,j} - \bv{B}\bv{\Pi}^{i,j} \|_\Fro > (3+3\epsilon/4)\cdot M \text{ for at least one } j\in \{1, \ldots, r\}.
        \end{align}
        Fixing any bad $\bv{B}$, 
        since $\bv{\Pi}^{i,j}$ has $\Omega\left(\log(\log|\F|/\delta)/\epsilon^2\right)$ columns, by \Cref{fact:hutch_standard}, $\|\bv{A}\bv{\Pi}^{i,j} - \bv{B}\bv{\Pi}^{i,j}\|_\Fro \geq (3+3\epsilon/4)\cdot M$ for any given $j$ with probability $\geq 1 - \delta/\log|\F|$.
        The sketching matrices $\bv{\Pi}^{i,1}, \ldots, \bv{\Pi}^{i,r}$ are drawn independently.
        Therefore,
        \[
        \Pr\left[\exists j : \|\bv{A}\bv{\Pi}^{i,j} - \bv{B}\bv{\Pi}^{i,j}\|_\Fro
        \geq (3+3\epsilon/4)\cdot M\right] \quad \geq 1 - (\delta / \log|\F|)^{r} \geq 1 - \delta / |\F|,
        \]
        where the last inequality follows from $r := \frac{\log|\F|}{\log \log |\F|}$.
        Since there are at most $|\F|$ bad matrices, by a union bound, \eqref{eq:bad_matrices_look_bad_after_triangle_inequality} holds for all bad $\bv{B}$ with probability at least $1 - \delta$. Thus \eqref{eq:bad_matrices_look_bad} holds, and by \eqref{eq:upper bound on output}, $\bv{\tilde B}$ cannot be bad and must satisfy $\|\bv{A} - \bv{\tilde B}\|_\Fro \leq (3+\epsilon)\cdot M$ as desired. Overall, this bound holds with probability $1-3\delta$. 
        The factor of 3 on $\delta$ comes from applying a union bound to conclude that \eqref{eq:bstar_never_removed}/\eqref{eq:norm_of_bstar}, \eqref{eq:left_sketch_guarantee}/\eqref{eq:error_of_r},
        and \eqref{eq:bad_matrices_look_bad}/\eqref{eq:bad_matrices_look_bad_after_triangle_inequality}
             all hold simultaneously.

        \paragraph{Guaranteed Termination.}
        \Cref{alg:two_sided_simulation} runs for at most $\sqrt{\log|\F|}$ iterations.
        We now argue that, with high probability, it returns a solution before then --- that is, the if-statement on \Cref{line:conditional} evaluates to true for some iteration $i \le \sqrt{\log |\F|}$ --- so \cref{line:fail} is never reached.
        To do this, we show that the size of $\mathcal{C}_i$ shrinks by a factor of $1/2^{\sqrt{\log|\F|}}$ at every iteration where the if-statement evaluates to false. Combined with \eqref{eq:bstar_never_removed}, which shows that the optimal solution $\bv{B}^*$ always remains in $\mathcal{C}_i$, we conclude that the number of times the if-statement evaluates to false before it evaluates to true (and an output is returned) can be at most $\sqrt{\log|\F|}$. 
        Define $\mathcal{C}_i[j]\subseteq \mathcal C_i$ as:
        \begin{align*}
            \mathcal{C}_i[j] = \left\{\bv{B} \in \mathcal{C}_{i} : \|\bv{A}\bv{\Pi}^{i,j} - \bv{B}\bv{\Pi}^{i,j}\|_\Fro \leq (1+\epsilon/12)\cdot M\right\}.
        \end{align*}
        That is, $\mathcal{C}_i[j]$
        is what $\mathcal{C}_{i+1}$ would be set to if we set $j_i^* = j$ at \cref{step:special j}. 
        Recalling that in \Cref{line:minCheck} we set $E_{i,j} = \min_{\bv{B}\in \mathcal{R}_{i}} \|\bv{\Psi}^\T\bv{A}\bv{\Pi}^{i,j} - \bv{\Psi}^\T\bv{B}\bv{\Pi}^{i,j}\|_\Fro$, we claim that the following holds with high probability:
        \begin{align}
            \label{eq:sampling_succeeds}
            \text{For all } i, j \text{ with } |\mathcal{C}_{i}[j]| \geq \frac{|\mathcal{C}_{i}|}{2^{\sqrt{\log|\F|}}}, \quad E_{i,j} \leq (1+\epsilon/6)\cdot M.
        \end{align}
     That is, if $\bv{\Pi}^{i,j}$ \emph{would not} significantly reduce the size of $\mathcal{C}_{i}$, then we will successfully find an accurate representative $\bv{R}^{i,j}$, and thus $\bv{\Pi}^{i,j}$ will not be selected on \Cref{step:special j}. 
     To prove \eqref{eq:sampling_succeeds}, we first prove that, with probability at least $1-\delta$,
     \begin{align}\label{eq:allSamplesSucceed}
        \text{For all } i, j \text{ with } |\mathcal{C}_{i}[j]| \geq \frac{|\mathcal{C}_{i}|}{2^{\sqrt{\log|\F|}}},\quad |\mathcal C_i[j] \cap \mathcal{R}_i| \ge 1.
     \end{align}
     That is, for any $i,j$ such that $\mathcal{C}_{i}[j]$ is large enough, we will successfully sample at least one representative from $\mathcal{C}_{i}[j]$ into $\mathcal{R}_{i}$. Since $\mathcal{R}_i$ contains $q = 2^{1.5 \cdot \sqrt{\log|\F|}\cdot \log(\log|\F|/\delta)}$ random elements of $\mathcal C_i$, assuming the condition of \eqref{eq:allSamplesSucceed} holds, some $\bv{B}\in \mathcal{C}_{i}[j]$ is included in $\mathcal{R}_{i}$ with probability at least:
       \begin{align*}
       1 - \left(1 - \frac{|\mathcal{C}_{i}[j]|}{|\mathcal{C}_{i}|}\right)^{q} \geq 1- \left(1 - \frac{1}{2^{\sqrt{\log|\F|}}}\right)^{q} \geq 1 - \frac{1}{e^{1.5 \cdot \log(\log|\F|/\delta)}} \geq 1- \frac{1}{\delta \log^{1.5}|\F|}. 
       \end{align*}
       Taking a union bound over  $O(\log^{1.5}|\F|)$ choices for $(i,j)$ yields \eqref{eq:allSamplesSucceed}. From \eqref{eq:left_sketch_guarantee}, we thus have
       \begin{align*}
       E_{i,j}
       \leq \|\bv{\Psi}^\T\bv{A}\bv{\Pi}^{i,j} - \bv{\Psi}^\T\bv{B}\bv{\Pi}^{i,j}\|_\Fro &\leq (1+\epsilon/16)\|\bv{A}\bv{\Pi}^{i,j} - \bv{B}\bv{\Pi}^{i,j}\|_\Fro \\ &\leq (1+\epsilon/16)\cdot (1+\epsilon/12)\cdot M \leq (1+\epsilon/6)\cdot M.
       \end{align*}
       for all such $i,j$.
       We conclude that \eqref{eq:sampling_succeeds} holds with probability at least $1-2\delta$, where the factor of two comes from union bounding over both \eqref{eq:left_sketch_guarantee}  and \eqref{eq:allSamplesSucceed}   holding.

       With \eqref{eq:sampling_succeeds} in place, we see that, whenever the if-condition on \Cref{line:conditional} 
       evaluates to false at iteration $i$, all $j$ for which $E_{i,j} > (1+\epsilon/6)\cdot M$ (and in particular, $j_i^*$), have $|\mathcal{C}_{i}[j]| < \frac{|\mathcal{C}_{i}|}{2^{\sqrt{\log|\F|}}}$. Thus, $|\mathcal C_{i+1}| = |\mathcal C_i[j_i^*]| \le |\mathcal{C}_{i}|/2^{\sqrt{\log|\F|}}$.
       It follows that, if the if-condition evaluates to false for all $i < z$, then $|\mathcal{C}_{z}| \leq |\F|\cdot (1/2)^{z\sqrt{\log|\F|}}$. 

       Now, suppose by way of contradiction that the if-condition evaluates to false for all $i \in \{0, \ldots, \sqrt{\log|\F|}\}$, causing us to reach \Cref{line:fail}. 
       Then we would have $\left|\mathcal{C}_{\sqrt{\log|\F|} + 1}\right| \leq |\F|\cdot (1/2)^{\left(\sqrt{\log|\F|} + 1\right)\sqrt{\log|\F|}} < 1$. 
       However, by \eqref{eq:bstar_never_removed}, $\left|\mathcal{C}_{\sqrt{\log|\F|} + 1}\right| \geq 1$ since this set contains $\bv{B}^*$. This is a contradiction. We conclude that the if-condition must evaluate to true for some $i \leq \sqrt{\log|\F|}$, and thus the algorithm must successfully return some $\bv {\tilde B}$.

       Thus, we finally conclude that, with probability at least $1-4\delta$, \Cref{alg:two_sided_simulation} returns a solution $\bv{\tilde B}$ satisfying $\|\bv{A} - \bv{\tilde B}\|_\Fro \leq (3+\epsilon)\cdot M$, as desired.

       \paragraph{Query Complexity.} We conclude the proof by calculating the query complexity of \Cref{alg:two_sided_simulation}. The method issues $m = O\left(\frac{1}{\epsilon^2}\sqrt{\log |\F|} \cdot \log(\log|\F|/\delta)\right)$ left queries to compute $\bv{W} = \bv{\Psi}^\T \bv{A}$. At every iteration of the for-loop, it issues at most $\ell = O\left(\frac{1}{\epsilon^2}\log(\log|\F|/\delta)\right)$ additional right queries to compute $\bv{Z} = \bv{A}\bv{\Pi}^{i,j}$. There are $\sqrt{\log |\F|} + 1$ iterations, so we conclude a final query complexity of $O\left(\frac{1}{\epsilon^2}\sqrt{\log |\F|}\log(\log|\F|/\delta)\right)$. 
    \end{proof}

\subsection{Estimating the Optimal Error}\label{sec:est}
We next show how to obtain an accurate upper bound $M \ge \min_{\bv B \in \F} \norm{\bv A - \bv B}_\Fro = \OPT$ for use in \Cref{alg:two_sided_simulation}.   Plugging this bound into  \Cref{thm:main_opt_known} will yield our main result, \Cref{thm:robust}.

\paragraph{Overview of Approach.}
Our approach starts by showing that a simple one-sided sketching algorithm using $O(\log(1/\delta))$ matvecs can efficiently find, with probability $> 1-\delta$, $M_{\rm init}$ satisfying the very coarse (but still multiplicative) error bound $\OPT \le M_{\rm init} \le 6 |\mathcal F| \cdot \OPT$. 
Starting with $M_{\rm init}$, we then run a binary search procedure to identify a more accurate upper bound on $\OPT$. 
In particular, we consider $O(\log|\F|/\epsilon)$ possible bounds: $M_i = \frac{(1+\epsilon)^i M_{\rm init}}{6|\F|}$ for $i \in \{0, 1, \ldots, \lceil\log_{(1+\epsilon)}(6|\F|) \rceil\}$. The smallest of these bounds is $\frac{M_{\rm init}}{6|\F|} \le \OPT$ and the largest is at least  $M_{\rm init} \ge \OPT$. Thus, at least one of them  must satisfy $\OPT \le M_i \le (1+\epsilon) \cdot \OPT$.

Furthermore, we can find this $M_i$ in $O(\log(\log|\F|/\epsilon))$ iterations of binary search, with \Cref{alg:two_sided_simulation} used to test whether a given bound is too large or too small. In particular, initialize an array containing all possible values of $M_i$. Elements of this array will be eliminated during the search process. Assume that at a step of the binary search, the median element of our current array is $M_z$. We can run \Cref{alg:two_sided_simulation} with $M_z$ as input to obtain an approximation $\bv{\tilde B}$. We can accurately estimate $\|\bv{A}-\bv{\tilde B}\|_\Fro$ using a small number of matvec queries via \Cref{fact:hutch_standard}. If we find that $\|\bv{A} - \bv{\tilde B}\|_\Fro \le (3+\epsilon) M_z$, then we know that either: Case 1) $M_z < \OPT$, but we have luckily found a solution with error $< (3+\epsilon) \cdot \OPT$ anyway, as desired, or Case 2): $M_z \ge \OPT$ is a valid upper bound and will proceed to find a tighter one.
Either way, we can safely remove all $M_j$ with $j > k$ from our binary search. Alternatively, if the algorithm returns $\bv{\tilde B}$ with $\|\bv{A} - \bv{\tilde B}\|_\Fro > (3+\epsilon) M_z$, then we know that with high probability, $M_z < \OPT$ and we can remove all $M_j$ with $j < k$ from our binary search. 

Overall, after running \Cref{alg:two_sided_simulation} $O(\log(\log|\F|/\epsilon))$ times, we will either have found some $\bv {\tilde B}$ with $\|\bv{A} - \bv{\tilde B}\|_\Fro \le (3+\epsilon) \cdot \OPT$, or we will have terminated our search on $M_z$ satisfying $\OPT \le M_i \le (1+\epsilon) \cdot \OPT$. Plugging $M_z$ into \Cref{alg:two_sided_simulation} will yield a solution $\bv{\tilde B}$ with $\|\bv{A} - \bv{\tilde B}\|_\Fro \le (3+\epsilon) M_z \le (3+5\epsilon) \cdot \OPT$. Adjusting $\epsilon$ by a constant factor gives a $(3+\epsilon)$ approximation and yields our main result, \Cref{thm:robust}.

\subsubsection{Coarse Approximation of \texorpdfstring{$\OPT$}{OPT}}
We first show that a coarse initial bound $M_{\rm init}$ can be obtained via a simple sketching algorithm (\Cref{alg:coarse}).

\begin{algorithm}[H]
    \caption{Coarse Approximation of $\OPT$}\label{alg:coarse}    
    \begin{algorithmic}[1]
        \Require Finite matrix family $\mathcal{F} \subset \mathbb{R}^{n \times n}$, failure probability $\delta \in (0,1)$.
        \Ensure $M$ such that, with probability at least $1-\delta$, $\OPT \le M \le 6|\mathcal F| \cdot \OPT$.
        \algrule
        \State Set $t = c \log(1/\delta)$ for a sufficiently large universal constant $c$.
        \State Let $\bv{\Pi}^1, \ldots, \bv{\Pi}^t \in \mathbb{R}^{n \times 2}$ be independent, each with i.i.d. $\mathcal{N}(0,1/2)$ Gaussian entries.\label{line:sketches} 

        \State For $i \in \{1,\ldots, t\}$, let $M_i = \sqrt{6|\mathcal F|} \cdot \min_{\bv B \in \mathcal{F}} \|\bv A \bv \Pi^i - \bv B \bv \Pi^i\|_\Fro$.\label{line:independentBounds}
    \State
        \Return $M = \mathrm{median}(M_1,\ldots,M_t)$.\label{line:median}
    \end{algorithmic}
\end{algorithm}  

The analysis of \Cref{alg:coarse} relies on a Frobenius-norm sketching guarantee, similar to \Cref{fact:hutch_standard}, but with a fixed constant sketch size and a very coarse approximation bound. In particular:
\begin{claim}[Coarse Frobenius-Norm Sketching]\label{clm:hutch_coarse}
    Let $\bv \Pi \in \mathbb{R}^{n \times 2}$ have independent $\mathcal{N}(0,1/2)$  Gaussian entries. For any $\delta \in (0,1)$ and any matrix $\bv C$ with $n$ columns, with probability $\ge 1-\delta$,
    \begin{align*}
        \sqrt{\frac{\delta}{2}} \|\bv C\|_\Fro \leq \|\bv{C \Pi}\|_\Fro \le \sqrt{\frac{2}{\delta}} \|\bv C\|_\Fro.
    \end{align*}
\end{claim}
\begin{proof}
    The upper bound follows directly from Markov's inequality: it is well-known that $\EE[\|\bv{C \Pi}\|_\Fro^2] = \|\bv C\|_\Fro^2$. Thus, $\Pr[\|\bv{C \Pi}\|_\Fro^2 > 2/\delta \cdot \|\bv C\|_\Fro^2] \le \delta/2$, and so $\Pr[\|\bv{C \Pi}\|_\Fro \le \sqrt{2/\delta} \|\bv C\|_\Fro] \ge 1-\delta/2$.
    The lower bound can be obtained in several ways. Concretely, Lemma 2.2 in \cite{PerssonCortinovisKressner:2022}, shows that, for any $\alpha \in (0,1)$, 
    \begin{align*}
       \Pr[\|\bv{C \Pi}\|_\Fro < \sqrt{\alpha} \|\bv C\|_\Fro] \le  \int_0^{\alpha} e^{-x} dx = 1-e^{-\alpha} \le \alpha.
    \end{align*}
    Setting $\alpha = \delta/2$, we have that $\sqrt{\frac{\delta}{2}} \| \bv C\|_\Fro \le \|\bv{C \Pi}\|_\Fro$ with probability at least $1-\delta/2$. 
    By a union bound, both the upper and lower bounds hold simultaneously with probability at least $1-\delta$, giving the claim.
\end{proof} 
Leveraging \Cref{clm:hutch_coarse}, we prove the following bound for \Cref{alg:coarse}:
\begin{claim}[Coarse Upper Bound on $\OPT$] For any finite matrix family $\F$, \Cref{alg:coarse} with input $\delta \in (0,1)$ uses $O(\log(1/\delta))$ matvec queries of the form $\bv x \mapsto \bv A \bv x$ and outputs $M$ that with probability $1-\delta$ satisfies:
    \begin{align*}
    \OPT &\le M \le 6|\F| \cdot \OPT & &\text{where} &\OPT &= \min_{\bv B \in \F} \norm{\bv A - \bv B}_\Fro.
    \end{align*}
\end{claim}
\begin{proof}
    We first analyze the accuracy of each of the random sketches $\bv{\Pi}^1,\ldots,\bv{\Pi}^t \in \R^{n \times 2}$ initialized in \Cref{line:sketches} of \Cref{alg:coarse}.
    Let $\bv \Pi \in \R^{n \times 2}$ be a Gaussian sketch with i.i.d. $\mathcal{N}(0,1/2)$  Gaussian entries. Applying \Cref{clm:hutch_coarse} with $\delta = 1/(3|\F|)$ and applying a union bound over all matrices of the form $\bv C = \bv A - \bv B$ for $\bv B \in \F$, we have that with probability at least $1-|\F| \cdot \delta = 2/3$, for all $\bv B \in \F$, $$1/\sqrt{6|\F|} \cdot \|\bv A - \bv B\|_\Fro \le \|\bv A \bv \Pi - \bv B \bv \Pi\|_\Fro \le \sqrt{6|\F|} \|\bv A - \bv B\|_\Fro.$$
 Thus, if we set $M = \sqrt{6 |\F|} \cdot \min_{\bv B \in \F} \|\bv A \bv \Pi - \bv B \bv \Pi\|_\Fro$, we have $\OPT \le M \le 6|\F| \cdot \OPT$ with probability $\ge 2/3$.
I.e., each $M_i$ in \Cref{line:independentBounds} satisfies $\OPT \le M_i \le 6|\F| \cdot \OPT$ with probability at least $2/3$. So in expectation, at least $2/3 \cdot t$ of these matrices satisfy the bound. By a standard Chernoff bound, since $t = c\log(1/\delta)$ for a large enough constant $c$, more than $t/2$ of the matrices $M_1,\ldots,M_t$ satisfy the desired bound with probability at least $1-\delta$. Thus, setting $M = \mathrm{median}(M_1, \ldots, M_t)$ in \Cref{line:median}, with probability $\ge 1-\delta$, we have $\OPT \le M \le 6|\F| \cdot \OPT$.
\end{proof}

\subsubsection{Improving the Bound via Binary Search}
We next show how to improve our initial bound $M_{\rm init}$ via a binary search procedure that uses \Cref{alg:two_sided_simulation} as a subroutine. This approach is formalized as \Cref{alg:binary_search} below.

\begin{algorithm}[H]
    \caption{Relative Error Approximation via Binary Search}\label{alg:binary_search}    
    \begin{algorithmic}[1]
        \Require Finite matrix family $\mathcal{F} \subset \mathbb{R}^{n \times n}$, failure probability $\delta \in (0,1)$, accuracy parameter $\epsilon \in (0,1)$.
        \Ensure Matrix $\bv{\tilde B} \in \mathcal{F}$ satisfying $\|\bv{A} - \bv{\tilde B}\|_\Fro \leq (3+\epsilon) \cdot \min_{\bv B \in \mathcal{F}} \|\bv A - \bv B\|_\Fro$ with probability at least $1-\delta$.
        \algrule
        \State Use \Cref{alg:coarse} to compute $M_{\rm init}$ satisfying $\OPT \le M_{\rm init} \le 6|\mathcal F| \cdot \OPT$ with probability at least $1-\delta/2$.\label{line:coarse}
        \State For  $i \in \{0, 1, \ldots, \lceil\log_{(1+\epsilon/12)}(6|\mathcal F|) \rceil\}$, set $M_i = \frac{(1+\epsilon/12)^i M_{\rm init}}{6|\mathcal F|}$.
        \State Let $\ell = \frac{c}{\epsilon^2} \cdot \log(\log|\F|/(\delta\epsilon))$ for a sufficiently large constant $c$. 
        \State Draw $\bv \Pi \in \{1/\sqrt{\ell}, -1/\sqrt{\ell}\}^{n \times \ell}$ with i.i.d. Rademacher entries.\label{line:camSketch1}
        \State Let $\bv W = \bv A \bv \Pi$.\label{line:camSketch}\Comment{Sketch used to estimate the error at each iteration of binary search.}
        \State Initialize $L = 0$, $R = \lceil\log_{(1+\epsilon/12)}(6|\mathcal F|) \rceil$, and $\bv{\tilde B} = \bv{0}$. \Comment{$\bv{\tilde B}$ tracks the output seen so far with the lowest error.}
        \While{L < R}
            \State Set $z = \lfloor (L + R)/2 \rfloor$.\label{line:binary_search_i}
            \State Let $\bv{\tilde B}_z$ be the output of \Cref{alg:two_sided_simulation} run with inputs $\delta' = \frac{\delta}{4 + 4R}$, \label{line:params} 
            $\epsilon' = \epsilon/24$, and  bound $M_z$.
            \If{$\|\bv{A} \bv{\Pi} - \bv{\tilde B}_z \bv{\Pi}\|_\Fro \leq (3+\epsilon/6) \cdot M_z$}\label{line:camCondition}
                \State $R = z$. 
                \Comment{Either $\OPT < M_z$ or else $\bv{\tilde B}$ is probably nearly optimal, so remove all larger bounds.}
                \State $\displaystyle \bv{\tilde B} = \argmin_{\bv{B} \in \{\bv{\tilde B},\bv{\tilde B}_z\}} \|\bv{A}\bv \Pi - \bv{B} \bv \Pi\|_\Fro$.
            \Else
                \State $L = z+1$.\Comment{With high probability $M_z < \OPT$, so remove all smaller bounds.}
            \EndIf
            \EndWhile
            \State \Return $\bv{\tilde B}$
    \end{algorithmic}
\end{algorithm}

We next analyze \Cref{alg:binary_search}, yielding our main result for finite family matrix approximation, restated below:
\begin{reptheorem}{thm:robust}
    For any finite family $\mathcal{F} \subset \R^{n \times n}$, \Cref{alg:binary_search} with inputs $\epsilon,\delta \in (0,1)$ returns $\bv{\tilde B} \in \mathcal{F}$ satisfying, with probability at least $1-\delta$,
    $$\|\bv{A} - \bv{\tilde B}\|_\Fro \leq (3+\epsilon) \cdot \min_{\bv B \in \mathcal{F}} \|\bv A - \bv B\|_\Fro.$$
    Further, the algorithm uses $O\left(\frac{\sqrt{\log|\F|} \cdot \log(\log|\F|/(\delta \epsilon)) \cdot \log(\log|\F|/\epsilon)}{\epsilon^2}\right) = \tilde O(\sqrt{\log|\F|}/\epsilon^2)$ matvecs with $\bv A$ and $\bv{A}^\transpose$.
\end{reptheorem}
\begin{proof}
    We claim that, with high probability, the following two bounds hold simultaneously for all $i \in \{0, \ldots, R\}$, where $\bv{\tilde B}_i$ is the output of \Cref{alg:two_sided_simulation} run with input $M_i$ and $\delta',\epsilon'$ as specified on \Cref{line:params}:
    \begin{align}
       (1-\epsilon/24) \|\bv A - \bv{\tilde B}_i\|_\Fro \le \|\bv A \bv \Pi - \bv{\tilde B}_i \bv \Pi\|_\Fro \le (1+\epsilon/24) \|\bv A - \bv{\tilde B}_i \|_\Fro. \label{eq:good_solution_if_good_bound_sketch},
    \end{align}
    and moreover,
    \begin{align}
        \text{if }M_i &\ge \OPT, & &\text{then} & \|\bv{A} - \bv{\tilde B}_i\|_\Fro &\leq (3+\epsilon/24) \cdot M_i. \label{eq:good_solution_if_good_bound}
    \end{align}
    Equation \eqref{eq:good_solution_if_good_bound_sketch} holds with probability at least $1-\delta/4$ by applying \Cref{fact:hutch_standard} with failure probability $\delta' = \delta/(4t)$ and sketch size $\ell = c \log(1/\delta')/\epsilon^2$ for a sufficiently large constant $c$, as set in \Cref{line:camSketch1}.
    Equation \eqref{eq:good_solution_if_good_bound} holds with probability at least $1-\delta/4$ by applying  \Cref{thm:main_opt_known} with failure probability $\delta' = \delta/(4(R+1))$ and union bounding over values of $i$.  By a union bound, both equations, along with the bound $\OPT \le M_{\rm init} \le 6 |\F| \cdot \OPT$ claimed on \Cref{line:coarse} hold with probability at least $1-\delta$. Going forward, we will assume all three of these bounds hold. 

    \paragraph{Accuracy.}

    Since $M_0 = \frac{M_{\rm init}}{6|\F|} \le \OPT$ and $M_t = (1+\epsilon/12)^{\lceil\log_{(1+\epsilon/12)}(6|\mathcal F|) \rceil} \cdot \frac{M_{\rm init}}{6|\F|} \ge M_{\rm init} \ge  \OPT$, we have that at least one of the bounds $M_i$ satisfies $\OPT \le M_i \le (1+\epsilon/12) \cdot \OPT$. Call this bound $M_{z^*}$. 
    
    \smallskip
    
    \noindent\textbf{Case 1: $\bv{\tilde B}_{z^*}$ is computed by the algorithm.} In this case, by \eqref{eq:good_solution_if_good_bound}, we have that $\|\bv{A} - \bv{\tilde B}_{z^*}\|_\Fro \le (3+\epsilon/24) \cdot M_{z^*} \le (3+\epsilon/24) \cdot (1+\epsilon/12) \cdot \OPT$. Further, since the output $\bv{\tilde B}$ minimizes $\|\bv A \bv \Pi - \bv{\tilde B} \bv \Pi\|_\Fro$ over all approximations computed by the algorithm, by \eqref{eq:good_solution_if_good_bound_sketch}, we have that
    \begin{align*}
        \|\bv A - \bv{\tilde B}\|_\Fro &\le \frac{1}{(1-\epsilon/24)} \cdot \|\bv A \bv \Pi - \bv{\tilde B} \bv \Pi\|_\Fro \le \frac{1}{(1-\epsilon/24)}\cdot \|\bv A \bv \Pi - \bv{\tilde B}_{z^*} \bv \Pi\|_\Fro \le \frac{(1+\epsilon/24)}{(1-\epsilon/24)} \|\bv A - \bv {\tilde B}_{z^*}\|_\Fro\\ &\le \frac{(1+\epsilon/24) \cdot (3+\epsilon/24)\cdot (1+\epsilon/12) \cdot \OPT}{(1-\epsilon/24)} \le (3+\epsilon) \cdot \OPT.
    \end{align*}
    \smallskip

    \noindent\textbf{Case 2: $\bv{\tilde B}_{z^*}$ is not computed by the algorithm.} In this case, $z^*$ must have been eliminated at some iteration of the while loop. This elimination could not happen at an iteration considering some $z > z^*$, since by \eqref{eq:good_solution_if_good_bound} and \eqref{eq:good_solution_if_good_bound_sketch}, we have for any such $z$, $\|\bv{A}\bv{\Pi} - \bv{\tilde B}_z \bv \Pi\|_\Fro \le (1+\epsilon/24) \cdot (3+\epsilon/24) \cdot M_z \le (3+\epsilon/6) \cdot M_z$, and thus the condition on \cref{line:camCondition} evaluates to true and $z^* < z$ is not eliminated. So, the elimination must happen at some iteration considering $z < z^*$. But for that to happen, \cref{line:camCondition} must evaluate to true at this iteration. So, we must have \[\|\bv A -\bv{\tilde B}_z\|_\Fro \le (3+\epsilon/6) \cdot M_z \le (3+\epsilon/6) \cdot M_{z^*} \le (3+\epsilon/6) \cdot (1+\epsilon/12) \cdot \OPT.\] Following the same argument as in Case 1, this gives that 
    \begin{align*}
        \|\bv A - \bv{\tilde B}\|_\Fro &\le \frac{(1+\epsilon/24)}{(1-\epsilon/24)} \cdot (3+\epsilon/6) \cdot (1+/12\epsilon) \cdot \OPT \le (3+\epsilon) \cdot \OPT.
    \end{align*}
     Thus, in both cases, we output a $(3+\epsilon)$ approximate solution, as desired. This bound holds as long as \eqref{eq:good_solution_if_good_bound_sketch},\eqref{eq:good_solution_if_good_bound} and the claimed bound on $M_{\rm init}$ hold simultaneously, which, as discussed, happens with probability at least $1-\delta$.

    \paragraph{Query Complexity.}  We use $\ell = O \left (\frac{\log(\log|\F|/(\delta\epsilon))}{\epsilon^2} \right)$ matvecs to compute the sketch $\bv W = \bv A \bv \Pi$ in \Cref{line:camSketch}, which is used for error estimation at each iteration.  Further, at each iteration of the while loop, we either set $R = z = \lfloor (L + R)/2 \rfloor$ or $L = z+1  = \lfloor (L + R)/2 \rfloor + 1$. In either case, $R-L$ decreases to at most $\frac{R-L}{2}$ after one iteration. Thus, the while loop runs for at most $1+\log_2(R) = O(\log(\log|\F|/\epsilon))$ iterations. In each iteration, by \Cref{thm:main_opt_known} we use $O\left(\frac{\sqrt{\log|\F|} \cdot \log(\log|\F|/\delta')}{(\epsilon')^2}\right)= O\left(\frac{\sqrt{\log|\F|} \cdot \log(\log|\F|/(\delta \epsilon))}{\epsilon^2}\right)$ matvecs to compute $\bv{\tilde B}_z$. Thus, our overall query complexity is $O\left(\frac{\sqrt{\log|\F|} \cdot \log(\log|\F|/(\delta \epsilon)) \cdot \log(\log|\F|/\epsilon)}{\epsilon^2}\right)$ matvecs, as claimed.\footnote{Note that we could improve this complexity by a $\log(\log|\F|/\epsilon)$ factor by sharing random queries across different calls to \Cref{alg:two_sided_simulation} and union bounding over the success probability of each call. However, for simplicity, we do not analyze this optimization here.}
\end{proof}

\section{Bounds for Infinite Families Via Covering}\label{sec:covering}

We next extend our results on finite family matrix approximation from  \Cref{sec:finite}  to infinite  families, which comprise most structured matrix families of interest in practice. We first state a generic result for any family $\F$ with finite \emph{covering number}, which follows essentially immediately from our main result on finite families, \Cref{thm:robust}. 

We then give an application to \emph{linear matrix families}, which are a common class of infinite families, including diagonal, tridiagonal, and banded matrices, other matrices with fixed sparsity patterns, Toeplitz matrices, Hankel matrices, and more. We show that a near-optimal approximation from a linear family of dimension $q$ can be obtained with $\tilde O(\sqrt{q})$ matvecs, improving on an $O(q)$ bound achievable via one-sided sketching techniques or vector-matrix-vector queries. It is not hard to see that both of these bounds are nearly optimal via the same lower bound approach used for finite families in \Cref{sec:lower_bounds}.

We start with a standard definition for the covering number of a matrix family $\mathcal{F}$, see e.g., \cite{Vershynin2018}.

        \begin{definition}[Covering Number]\label{def:cover}
            Consider a family of matrices $\F \subset \R^{n \times n}$. We say that $\mathcal{F}' \subseteq \mathcal{F}$ is an \emph{$\alpha$-cover} of $\mathcal{F}$ if, for all $\bv{X} \in \mathcal{F}$, there exists a $\bv{Y} \in \mathcal{F}'$ such that $\|\bv{X} - \bv{Y}\|_\Fro \leq \alpha$. The \emph{$\alpha$-covering number} of $\mathcal{F}$ is the minimum cardinality of any $\alpha$-cover of $\mathcal{F}$.
        \end{definition}

        We remark that an $\alpha$-cover is also frequently referred to as an $\alpha$-net in the literature. With this definition in hand, we are ready to give our result for infinite families with bounded covering numbers.

        \begin{theorem}\label{thm:covering-number-approximation}
            Consider any $\alpha > 0$. There is an algorithm that, for any matrix family $\mathcal{F} \subset \mathbb{R}^{n \times n}$ with a given $\alpha$-cover of cardinality $\Gamma_{\alpha}$, and any $\epsilon,\delta \in (0,1)$ uses  $O\left(\frac{\sqrt{\log\Gamma_\alpha} \cdot \log(\log\Gamma_\alpha/(\delta \epsilon)) \cdot \log(\log \Gamma_\alpha/\epsilon)}{\epsilon^2}\right) = \tilde O\left(\frac{\sqrt{\log \Gamma_{\alpha}}}{\epsilon^2} \right)$ matrix-vector product queries with $\bv A$ and $\bv A^\T$ and returns $\bv{\tilde B} \in \mathcal{F}$ satisfying, with probability at least $1 - \delta$,
            \begin{align*}
                \|\bv{A} - \bv{\tilde B}\|_\Fro \leq (3+\epsilon)\cdot \inf_{\bv{B}\in \mathcal{F}}\|\bv{A} - \bv{B}\|_\Fro + (3 + \epsilon) \alpha.
            \end{align*}
        \end{theorem}
Note that the bound of \Cref{thm:covering-number-approximation} is stated with $\inf_{\bv B \in \mathcal{F}} \|\bv A - \bv B\|_\Fro$ instead of $\min_{\bv B \in \mathcal{F}} \|\bv A - \bv B\|_\Fro$, since a minimizer may not exist for general infinite families. Additionally, note that the theorem assumes that an $\alpha$-cover of $\mathcal{F}$ is given as input. This is required for the algorithm to be constructive. For typical families encountered in practice, efficient coverings can be constructed using standard discretization techniques. 
        \begin{proof}
        Let $\mathcal{F}'$ be the given $\alpha$-cover of $\mathcal{F}$ (\Cref{def:cover}) with $|\mathcal{F}'| = \Gamma_{\alpha}$. 
        For $t > 0$, let $\bv{B}_t \in \mathcal{F}$ be any matrix satisfying $\|\bv{A} - \bv{B}_t\|_\Fro \leq \inf\limits_{\bv{B} \in \mathcal{F}}\|\bv{A} - \bv{B}\|_\Fro + t$ and let $\bv{B}_t' \in \mathcal{F}'$ be so that $\|\bv{B}_t' -\bv{B}_t\|_\Fro \leq \alpha$. 
        By \Cref{thm:robust}, there is an algorithm that issues $O\left(\frac{\sqrt{\log\Gamma_\alpha} \cdot \log(\log\Gamma_\alpha/(\delta \epsilon)) \cdot \log(\log \Gamma_\alpha/\epsilon)}{\epsilon^2}\right)$ matvec queries and returns $\bv{\tilde B}$ satisfying $\|\bv{A} - \bv{\tilde {B}}\|_\Fro \leq (3+\epsilon)\min\limits_{\bv{B} \in \mathcal{F}'}\|\bv{A} - \bv{B}\|_\Fro$ with probability at least $1-\delta$. Applying the triangle inequality yields
        \begin{align*}
            \|\bv{A} - \bv{\tilde B}\|_\Fro \leq (3+\epsilon)\min\limits_{\bv{B} \in \mathcal{F}'}\|\bv{A} - \bv{B}\|_\Fro &\leq (3+\epsilon)\|\bv{A} - \bv{B}_t'\|_\Fro\\
            & \leq (3+\epsilon)\|\bv{A} - \bv{B}_t\|_\Fro + (3+\epsilon)\|\bv{B}_t' - \bv{B}_t\|_\Fro\\
            & \leq (3+\epsilon)\inf\limits_{\bv{B} \in \F} \|\bv{A} - \bv{B}\|_\Fro + (3+\epsilon)t + (3+\epsilon) \alpha.
        \end{align*}
        Letting $t \to 0$ yields the desired result. 
        \end{proof}

        We next apply \Cref{thm:covering-number-approximation} to show that a near-optimal approximation from any linear matrix family of dimension $q$ can be computed using $\tilde O(\sqrt{q})$ matvec queries. We first formally define a linear matrix family.

\begin{definition}[Linear Matrix Family]\label{def:linear}
    $\mathcal{L} \subseteq \mathbb{R}^{n \times n}$ 
    is a linear matrix family of dimension $q$ if there exist $q$ linearly independent matrices $\bv P_1,\ldots,\bv P_q \in \mathbb{R}^{n \times n}$ such that $\mathcal{L} = \{\bv B : \bv B = \sum_{i=1}^q c_i \bv P_i\text{ for some }c_1,\ldots,c_q \in \mathbb{R}\}$. I.e., $\mathcal{L}$ is the set of all matrices that can be written as linear combinations of $\bv P_1,\ldots, \bv P_q$.
\end{definition}

        \begin{repcorollary}{cor:linear}[Linear Family Approximation]
            For any linear matrix family, $\mathcal{L} \subseteq \mathbb{R}^{n \times n}$, with dimension $q$ and any $\alpha,\epsilon,\delta \in (0,1)$, it there is an algorithm that outputs $\bv{\tilde{B}} \in \mathcal{L}$ satisfying, with probability $\ge 1-\delta$,
            \begin{align*}
                \|\bv{A} - \bv{\tilde{B}}\|_\Fro \leq (3+\epsilon)\min_{\bv{B}\in \mathcal{L}}\|\bv{A} - \bv{B}\|_\Fro + \alpha\|\bv{A}\|_\Fro,
            \end{align*}
             using $O\left(\frac{\sqrt{q \cdot \log(1/\alpha)} \cdot \log(\log(q \log(1/\alpha)/(\delta \epsilon)) \cdot \log(\log(q \log(1/\alpha)/\epsilon)}{\epsilon^2}\right) = \tilde O \left (\frac{\sqrt{q \log(1/\alpha)}}{\epsilon^2} \right )$ matvec queries with $\bv A$ and $\bv A^\T$.
        \end{repcorollary}
        Note that, unlike \Cref{thm:robust} for finite families, \Cref{cor:linear} does not give a pure relative error guarantee, but also has additive error $\alpha \|\bv A\|_\Fro$. However, the leading dependence on $\alpha$ in the query complexity is just $O(\sqrt{\log(1/\alpha)})$. Thus, $\alpha$ can be set very small, and the additive error can be made negligible, comparable to round-off error if $\bv{\tilde B}$ is represented in finite precision. Nevertheless, achieving a pure relative error guarantee is an interesting open problem.

        \begin{proof}
        Since $\mathcal{L}$ is unbounded, it has an unbounded covering number, and thus we cannot apply \Cref{thm:covering-number-approximation} directly. Instead, we will apply the theorem to $\mathcal{L}$ intersected with a ball of radius $O(\|\bv{A}\|_\Fro)$ around the origin. It is not hard to argue that this ball must contain the optimal approximation to $\bv A$ in $\mathcal{L}$, and thus we can restrict our search to this bounded family, which we will argue has bounded covering number, allowing us to apply \Cref{thm:covering-number-approximation}.

        Let $\bv{B^*} = \argmin_{\bv{B} \in \mathcal{L}} \|\bv{A} - \bv{B}\|_\Fro$. Since $\mathcal L$ is closed and convex, such a minimizer exists and is unique. Further, since $\bv{0} \in \mathcal{L}$, $
            \|\bv{A} - \bv{B^*}\|_\Fro \leq \|\bv{A} - \bm{0}\|_\Fro = \|\bv{A}\|_\Fro$. Thus, by triangle inequality, $\|\bv B^*\|_\Fro \le 2 \|\bv A \|_\Fro$ -- i.e., $\bv B^*$ is contained $\mathcal{L} \cap \mathcal{B}(2\|\bv A\|_\Fro)$, where $\mathcal{B}(2\|\bv A\|_\Fro)$ is the Frobenius norm ball of radius $2\|\bv A\|_\Fro$ centered at the origin.

        We can use sketching to estimate $\|\bv{A}\|_{\Fro}$ and thus find a small ball containing $\mathcal{B}(2\|\bv A\|_\Fro)$. Let $\bv{\Pi} \in \left\{\frac{-1}{\sqrt{\ell}}, \frac{1}{\sqrt{\ell}}\right\}^{n \times \ell}$ have independent Rademacher entries and $\ell = O(\log(1/\delta))$ columns. By \cref{fact:hutch_standard},  with probability at least $1 - \delta/2$, \begin{align}\label{eq:camAPi}\frac12\|\bv{A}\|_\Fro \leq \|\bv{A\Pi}\|_\Fro \leq \frac32 \|\bv{A}\|_\Fro.
        \end{align}
        Thus, $4 \|\bv A \bv \Pi \|_\Fro \ge 2\|\bv A\|_\Fro$, and
        we have $\bv B^* \in \mathcal G$ where $\mathcal G = \mathcal{L} \cap \mathcal{B}(4\|\bv{A\Pi}\|_\Fro)$. 

        Thus, it suffices to find a near-optimal approximation to $\bv A$ from $\mathcal G$. We do so using \Cref{thm:covering-number-approximation}, after arguing that $\mathcal G$ has bounded covering number. Let $\alpha' = \frac16\alpha\|\bv{A\Pi}\|_\Fro$. It is well known (c.f. \cite[Corollary 4.2.13]{Vershynin2018}) that, since $\mathcal G$ is nothing more than a $q$-dimensional ball, we can construct an $\alpha'$-cover $\mathcal G'$ of $\mathcal G$ whose size is bounded by\footnote{One can construct a cover for $\mathcal{G}$ by first orthonormalizing $\bv{P}_1,\ldots,\bv{P}_q$, in which case $\mathcal{G}=\{c_1\bv{P}_1 + \cdots + c_q \bv{P}_q: \left(c_1^2 + \cdots + c_q^2\right)^{1/2} \leq 4\|\bv{A}\|_\Fro\}$. A cover for $\mathcal{G}$ then follows from a cover of the Euclidean ball in $\mathbb{R}^q$ with radius $4\|\bv{A}\|_{\Fro}$.}
        \begin{equation}\label{eq:covernumber}
            |\mathcal G'| \leq \left(\frac{2 \cdot 4\|\bv{A\Pi}\|_\Fro}{\alpha'} + 1\right)^q = \left(\frac{48}{\alpha} + 1\right)^q.
        \end{equation}
        Thus, by \Cref{thm:covering-number-approximation}, we can issue $O\left(\frac{\sqrt{\log |\mathcal G'|} \cdot \log(\log |\mathcal G'|/(\delta \epsilon)) \cdot \log(\log |\mathcal G'|/\epsilon)}{\epsilon^2}\right)$ matvecs to obtain an approximation $\bv{\tilde B} \in \mathcal G$ such that, with probability at least $1-\delta/2$,
        \begin{align*}
            \|\bv A - \bv{\tilde B}\|_\Fro
            &\leq (3+\epsilon)\min_{\bv{B}\in \mathcal{G}}\|\bv{A}-\bv{B}\|_\Fro + (3 + \epsilon) \alpha'\\
            &= (3+\epsilon)\min_{\bv{B}\in \mathcal{L}}\|\bv{A} - \bv{B}\|_\Fro + (3 + \epsilon) \frac2{4 \cdot 3} \alpha \|\bv{A\Pi}\|_\Fro\\
            &\leq (3+\epsilon)\min_{\bv{B}\in \mathcal{L}}\|\bv{A} - \bv{B}\|_\Fro + \frac23 \alpha \|\bv{A\Pi}\|_\Fro
        \end{align*}
        Recalling from \eqref{eq:camAPi} that $\|\bv{A\Pi}\|_\Fro \leq \frac32 \|\bv{A}\|_\Fro$ and union bounding over the failure probability of \eqref{eq:camAPi} and \Cref{thm:covering-number-approximation}, the claimed bound on $\|\bv{A} - \bv{\tilde B}\|_\Fro$ holds with probability at least $1 - \delta$. The query complexity bound follows by observing that $\log|\mathcal G'| = O(q \cdot \log(1/\alpha))$.
    \end{proof}

\section{Lower bound}
\label{sec:lower_bounds}

We now argue that the dependence on $\log |\mathcal{F}|$ in \Cref{thm:robust} is essentially tight: there exist finite families where $\Omega(\sqrt{\log |\mathcal{F}|})$ queries are required to solve \Cref{prob:intro} even to a coarse approximation factor. We begin by stating a more detailed version of \Cref{thm:two_sided_lower_bound} below:
\begin{reptheorem}{thm:two_sided_lower_bound}
    For any $n > 0$, any approximation factor $\gamma \ge 1.1$, and any $k \leq c'n^2 \log \gamma$ for a universal constant $c'$, there 
    is a finite family of matrices $\F \subseteq \R^{n\times n}$ with $\log |\F| = O(k)$ and a distribution over  $\bv A \in \R^{n \times n}$ such that, if $\bv{\tilde B}$ is the output of any algorithm that makes fewer than $\frac{c\sqrt{\log|\F|}}{\sqrt{\log \gamma}}$ matvec queries to $\bv A$ and $\bv A^\T$ for a universal constant $c > 0$,
\[ \|\bv A- \bv{\tilde B}\|_\Fro \geq \gamma \cdot \min_{\bv B \in \F} \|\bv A - \bv B\|_\Fro, \]
 with probability $\geq 9/10$ (over the randomness in $\bv A$ and possibly the algorithm).
In particular, $\Omega(\sqrt{\log |\F|})$ matvec queries are necessary to achieve any constant approximation factor $\gamma$ with probability $\ge 1/10$.
\end{reptheorem}
Note that the dependence on $1/\sqrt{\log \gamma}$ in \Cref{thm:two_sided_lower_bound} is optimal: in \Cref{clm:hutch_coarse}, we show how to find a $\gamma = O(|\F|)$-approximate solution using just $O(1) = O(\sqrt{\log |\F|/\log(\gamma)})$ matvec queries.

Before proving \cref{thm:two_sided_lower_bound}, we first present several technical lemmas that will be used in the argument. As discussed in \Cref{sec:contributions}, the idea behind the lower bound is simple: we let $\mathcal F$  be a sufficiently fine net over all matrices that are non-zero on the top $z\times z$ block of $\bv{A}$ and zeros elsewhere. We will have $\log |\F| = O(z^2)$. To solve \Cref{prob:intro}, we essentially must recover all of $\bv A$'s top $z\times z$ block, which requires $\Omega(z) = \Omega(\sqrt{\log|\F|})$ queries. 

To formalize this argument, we need to chose a distribution over inputs $\bv A$ that is sufficiently anti-concentrated so that its best approximation in $\F$ is hard to learn. There are many possible options here -- we will let $\bv A$ be a Wishart matrix (i.e., $\bv A = \bv G \bv G^\T$ where $\bv G$ has i.i.d. Gaussian entries), since the distribution of these matrices after querying with a set of possibly adaptively chosen matvecs is well understood. In particular, we will use the follow result: 

\begin{lemma}[{\protect\textit{Lemma 13} in \cite{BravermanHazanSimchowitz:2020}}]\label{lem:wishart_query}
    Let $\bv A= \bv G \bv G^\T$, where $\bv{G} \in \mathbb{R}^{n \times n}$ has i.i.d. $\mathcal{N}(0, 1)$ entries. For any $m \leq n$, consider any set of adaptively chosen queries $\bv{q}_1, \ldots, \bv{q}_m \in \R^{n}$, where $\bv{q}_i$ is drawn from a distribution $\mathcal{D}_i$ that is a function of $\bv{q}_1, \ldots, \bv{q}_{i-1}$ and the responses $\bv{A} \bv q_1\ldots,\bv A \bv q_{i-1}$. We can construct a matrix $\bv{V} \in \mathbb{R}^{n\times n-m}$ with orthonormal columns and a matrix $\bv{\Delta}\in \mathbb{R}^{n\times n}$, both of which are functions only of $\bv{q}_1,\ldots,\bv q_m$ and $\bv A \bv q_1,\ldots,\bv A \bv q_m$,
    such that
    \[
    \bv A= \bv{\Delta}+\bv{V}\bv{W}\bv{V}^\T,
    \]
    where 
    $\bv{W}=\widehat{\bv{G}}\widehat{\bv{G}}^\T$ for $\widehat{\bv{G}} \in \mathbb{R}^{(n-m)\times(n-m)}$ with  i.i.d. $ \mathcal N(0,1)$ entries independent of $\bv{q}_1,\ldots,\bv q_m$ and $\bv A \bv q_1,\ldots,\bv A \bv q_m$.
\end{lemma}

We will apply the following bound to argue that after querying with a sufficiently small set of matvecs, $\bv A$ will still have fairly anti-concentrated entries, so that finding its best approximation in $\F$ is difficult. In particular, we use that, since the entries of $\bv W$ from \Cref{lem:wishart_query} are the inner products of the Gaussian rows of $\widehat{\bv G}$, they are anti-concentrated.

\begin{lemma}\label{lem:gausian-inner-prod} There exists a constant $c >0$ such that if $\bv x,\bv y$ each have i.i.d. $\mathcal N(0, 1)$ entries, and either $\bv x = \bv y$ or $\bv x$ and $\bv y$ are independent,
    then for any $\epsilon>0$ and $t\in\R$, if $n > c/ \epsilon^2$,
        \[
        \Pr\Bigl[ \big| \bv{x}^\T\bv{y} - t \big| < \epsilon \sqrt{n} \Bigr] < \epsilon.
    \]
\end{lemma}
\Cref{lem:gausian-inner-prod} can be proven using standard anti-concentration bounds for Gaussians and concenration bounds for Chi-Squared random variables.
It also follows as a special case of the more general Lemma 1 of \cite{AmselChenHalikias:2024} by setting $\bv u$ and $\bv v$ in that lemma to be standard basis vectors. 

We are now ready to prove \Cref{thm:two_sided_lower_bound}.

\begin{proof}[Proof of \Cref{thm:two_sided_lower_bound}] 
We focus on proving the result for $k^2 = c'n^2\log n$, as this case easily generalizes to smaller values of $k^2$ by considering the same hard instance on the top $z \times z$ principal submatrix of $\bv A$ for some $z \le n$, and letting all other entries outside this submatrix equal $0$.

Let $\alpha= \frac{n^{1.5}}{160 \cdot \gamma}$.
Let $\bv A= \bv G \bv G^\T$, where each entry of $\bv{G} \in \mathbb{R}^{n\times n}$ is drawn i.i.d. from $\mathcal{N}(0, 1)$ and let $\F$ be an $\alpha$-cover (\cref{def:cover}) of the ball $\mathcal{B}(c_1 n^{1.5}) = \{\bv X \in \R^{n\times n} : \|\bv X\|_{\Fro} \le c_1 n^{1.5}\}$ for a constant $c_1$ to be fixed shortly. It is well known (c.f. \cite[Corollary 4.2.13]{Vershynin2018}) that we can construct $\mathcal F$ with $|\mathcal F| = \left (\frac{2 c_1n^{1.5}}{\alpha} +1\right )^{n^2}$ and thus $\sqrt{\log|\F|} = O(n \cdot \sqrt{\log \gamma})$. 

\smallskip

\noindent \textbf{Upper bound on the optimal error.}
An elementary computation reveals that $\E[\|\bv G\bv G^\T\|_\Fro^2] = 2n^3 + n^2 $.
Hence, by Markov's inequality, for some constant $c_1$, $\|\bv G\bv G^\T\|_\Fro \leq c_1n^{1.5}$ with probability at least $19/20$; i.e. $\bv A=\bv G\bv G^\T$ lies within the Frobenius ball $\mathcal{B}(c_1 n^{1.5})$.
Assuming this event holds, then by the definition of an $\alpha$-cover:
\begin{align}\label{eq:upper2}
    \min_{\bv B \in \F} \|\bv A - \bv B\|_\Fro \le \alpha = \frac{n^{1.5}}{160 \cdot \gamma}.
\end{align}

\smallskip

\noindent \textbf{Lower bound on algorithm's error.}
By \Cref{lem:wishart_query}, after $m$ adaptive queries, we can construct $\bv{\Delta}\in \mathbb{R}^{n\times n}$ and $\bv{V} \in \mathbb{R}^{n\times n-m}$ with orthonormal columns such that
    $\bv A= \bv{\Delta}+\bv{V}\bv{W}\bv{V}^\T,$
    where $\bv{W}=\widehat{\bv{G}}\widehat{\bv{G}}^\T$ with each entry of $\widehat{\bv{G}} \in \mathbb{R}^{(n-m) \times (n-m)}$ is i.i.d. $ \mathcal N(0,1)$ and is independent of $\bv{q}_1,\ldots,\bv q_m$ and $\bv A \bv q_1,\ldots,\bv A \bv q_m$. Thus, $\bv W$ is independent of the output $\bv{\tilde B}$ of the algorithm.
    
For any fixed $\bv{\tilde B} \in \R^{n \times n}$, let $\bv{T}=\bv{\tilde B}-\bv{\Delta}$. Then, $\|\bv A - \bv{\tilde B}\|_\Fro = \| \bv{V}\bv{W}\bv{V}^\T -\bv T \|_\Fro$.
Let $\bv Z= \bv V^\T \bv T \bv V $. We have:
\begin{align*}
\|\bv A - \bv{\tilde B}\|_\Fro =
    \|\bv{V}\bv{W}\bv{V}^\T -\bv T \|_\Fro \geq \| \bv{V}\bv{W}\bv{V}^\T -\bv V \bv V^\T \bv T \bv V  \bv V^\T\|_\Fro = \| \bv{V}\bv{W}\bv{V}^\T -\bv V \bv Z \bv V^\T\|_\Fro
    = \|\bv W - \bv Z \|_\Fro.
\end{align*}
The inequality uses that $\bv V$ has orthonormal columns. It follows that:
\begin{align}\label{eq:camTiredBound}
    \Pr \left [\|\bv A - \bv{\tilde B}\|_\Fro \geq \frac{n^{1.5}}{160} \right ] \geq \Pr \left [\| \bv{W} -\bv Z \|_\Fro \geq \frac{n^{1.5}}{160} \right].
\end{align}
For any constant $c_2 > 0$, we can choose $n > 2 c_2$, and since by assumption, $m \le \frac{c \sqrt{\log |\F|}}{\sqrt{\log \gamma}} \le \frac{n}{2}$, if we set $c$ small enough, we will have $n-m > c_2$. Setting $c_2$ large enough we can apply \cref{lem:gausian-inner-prod} with $\epsilon = 1/40$ to each entry $W_{ij}$ of $\bv W$. $\bv W = \widehat{\bv G} \widehat{\bv G}^\T$, so $W_{ij} = \bv x^\T \bv y$ for Gaussian vectors $\bv x,\bv y$ that are either independent (when $i \neq j$) or equal (when $i =j$). Thus, \Cref{lem:gausian-inner-prod} applies and we have for all $i,j \in \{1,\ldots,n-m\} \times \{1,\ldots n-m\}$
\begin{align*}
    \Pr\Bigl[ \big| W_{ij} - Z_{ij} \big| < \frac{\sqrt{n-m}}{40} \Bigr] < \frac{1}{40}.
\end{align*}
Thus, by Markov's inequality, 
\begin{align*}
        \Pr\Biggl[ \sum_{i=1}^{n -m} \sum_{j=1}^{n-m} \ones\Big[| W_{ij} - Z_{ij} | < \frac{\sqrt{n-m}}{40} \Big] 
        \geq \frac{(n-m)^2}{2} \Biggr] 
        \leq \frac{1/40}{1/2} =\frac{1}{20}.
\end{align*}
So, with probability $\ge \frac{19}{20}$, there are at least $\frac{(n-m)^2}{2}$ pairs $(i,j)$ such that $| W_{ij} - Z_{ij} |^2 \geq \frac{({n-m})}{40^2}.$
Thus,
\begin{align}\label{eq:tired2}
    \Pr\bigg[ \big\| \bv{W} -\bv Z \big\|_\Fro^2 \geq \frac{(n-m)^3}{2 \cdot 40^2}\bigg] \geq \frac{19}{20}.
\end{align}
Recalling that we assume $m \le \frac{c \sqrt{\log |\F|}}{\sqrt{\log \gamma}}$ for small enough $c$ and thus $m \le \frac{n}{2}$, we have $\frac{(n-m)^3}{2 \cdot 40^2} \ge \frac{n^3}{160^2}$. Applying  \eqref{eq:camTiredBound} and plugging back into \eqref{eq:tired2}, we thus have:
\begin{align}\label{eq:lower2}
    \Pr\left [\|\bv A - \bv{\tilde B}\|_\Fro \geq \frac{n^{1.5}}{160} \right ] \geq \Pr\left [\|\bv W - \bv{Z}\|_\Fro^2 \geq \frac{n^{3}}{160^2}\right ] \ge \frac{19}{20}.
\end{align}
Combining equations \cref{eq:upper2} and \cref{eq:lower2} and applying a union bound over both holding, we conclude that with probability at least $\frac{9}{10}$,
    $\|\bv A - \bv{\tilde B}\|_\Fro \geq \gamma \cdot \min_{B\in \mathcal{F}}\|\bv A - \bv B\|_\Fro$, 
as desired.
\end{proof}

\paragraph{Application to Butterfly Matrices.}
While the hard instance of \Cref{thm:two_sided_lower_bound} may seem somewhat artificial, it is actually closely related to structured matrix classes that arise commonly in applications. In particular, the bound easily implies a lower bound of $\Omega(\sqrt{n})$ matvecs for approximation from the class of constant-rank butterfly matrices, which have been considered in many recent works \cite{LiYangMartin:2015,LiYang:2017,LiuXingGuo:2021}. This lower bound matches the best known upper bound upper bound of \cite{LiYang:2017} up a $\log n$ factor.

\begin{corollary}\label{cor:butterfly}
    For $n = 2^k$ for integer $k > 0$, let $\mathcal{F}\subset \R^{n \times n}$ be the class of rank-$1$ butterfly matrices \cite{LiYangMartin:2015}. For any $\gamma > 1.1$, solving \Cref{prob:intro}  for $\F$ to accuracy $\gamma$ for any $\gamma \ge 1.1$ with probability $\ge 1/10$ requires $\Omega(\sqrt{n}/\sqrt{\log \gamma})$ queries.
\end{corollary}
We refer the reader to \cite{LiYang:2017} for a complete definition of butteryfly matrices. We use the fact that rank-$1$ butterly matrices are equivalent to the family of matrices with \emph{complimentary rank-$1$ structure}. In short, a matrix $\bv{B}$ has complimentary rank-$1$ if, for all $i,j < \log_2(n)$, given a partition of $\bv{B}$ into a grid of blocks of size $2^i\times 2^j$, every block has rank $1$. With this definition in place, we can prove the claim:
\begin{proof}
    Consider the family $\mathcal{F'} \subset \R^{n \times n}$ where the top-left entry of every block in a partition of the matrix into $n$, $\sqrt{n} \times \sqrt{n}$ blocks is arbitrary, and all other entries are zero. It is not hard to see that any matrix in this family has complimentary rank-$1$ since any block in a partition into $2^i\times 2^j$ blocks has at most a single non-zero column.  I.e., $\F'$ is a subset of the family of rank-$1$ butterfly matrices, $\mathcal F$.

    Moreover, it is not hard to see that, up to a permutation of row and column indices, this family is equivalent to the set of matrices that take arbitrary non-zero values on their top $\sqrt{n} \times \sqrt{n}$ principal submatrix and are $0$ anywhere else. Now, let $\bv A \in \F'$ be an appropriate row/column permutation of the hard case of \Cref{thm:two_sided_lower_bound}, a matrix  whose top $\sqrt{n} \times \sqrt{n}$ block is a Wishart matrix $\bv G \bv G^\T$ for $\bv G \in \R^{\sqrt n \times \sqrt{n}}$. Any algorithm for finding an optimal approximation to $\bv A$ in $\F$ must exactly recover $\bv A$ since $\bv A \in \mathcal F' \subseteq \F$. Thus, any such algorithm yields an algorithm, with the same query complexity, for finding an optimal approximation to $\bv A$ in any other family, including in the finite $\alpha$-cover considered in \Cref{thm:two_sided_lower_bound}. Thus, such an algorithm must make $\Omega(\sqrt{n}/\sqrt{\log \gamma})$ queries, yielding the theorem.
\end{proof}

\section{Open Problems}
\label{sec:open_problems}
We conclude by highlighting a few open problems that would be interesting to address in follow-up work. First, for finite families, we conjecture that it is possible to obtain a $(1+\epsilon)$ approximation with $\tilde{O}(\sqrt{\log |\F|}/\poly(\epsilon))$ queries, improving on the $(3+\epsilon)$ bound achieved by \Cref{thm:robust}. We would also like to understand the roll of adaptivity -- is it possible that the same bound is achievable even if the query vectors $\bv{x}_1, \ldots, \bv{x}_m$ are chosen all at once, or with just a few rounds of adaptivity? If not, can we prove a lower bound? Our current algorithm issues a non-adaptive set of left queries to $\bv{A}^\transpose$, but our right queries to $\bv{A}$ are chosen adaptively. Currently, the best known non-adaptive method (discussed in \Cref{sec:finite}) requires $\tilde {O}(\log |\F|/\epsilon^2)$ queries; that is it takes no advantage of having access to $\bv A^\T$.

We are also interested in obtaining improved bounds for linearly parameterized families. We conjecture that $\tilde{O}(\sqrt{q}/\poly(\epsilon))$ queries suffice to achieve a pure $(1+\epsilon)$ relative error approximation for any linearly parameterized family with dimension $q$. Eliminating the additive $\alpha\cdot \|\bv{A}\|_\Fro$ from our \Cref{cor:linear} will seemingly require a different approach that does not rely on a black-box reduction to finite families. 

Finally, beyond query complexity, it is interesting to consider computational complexity. The runtime of our current algorithms scale as $\tilde{O}(|\mathcal{F}|)$. While this is essentially optimal if we do not assume any additional structure about $\mathcal{F}$ (we must at least read the $|\mathcal{F}|$ matrices in the family), we could hope to do significantly better if $\mathcal{F}$ has more structure -- e.g., has a compact parameterization, or is a discretization of a simply described infinite family. For example, a concrete question is if it is possible to match the $\tilde{O}(\sqrt{q})$ query complexity of \Cref{cor:linear} for dimension $q$ linear families with an algorithm that runs in $\poly(n,q)$ time.

\section*{Acknowledgements} This work was partially supported by NSF awards 2427362, 2046235, and 2045590. We would like to thank Diana Halikias, Alex Townsend, Nicholas Boullé, Sam Otto, Rikhav Shah, and Paul Beckman for helpful conversations that lead to this work. We thank Raphael Meyer for helpful comments. We would also like to thank Leon Zhao and David Woodruff for helpful conversations after the initial publication. 

\bibliographystyle{apalike}
\bibliography{refs.bib}

\appendix

\section{Finite Family Approximation in Weaker Query Models}\label{sec:weakerModels}

As discussed, our $\tilde O(\sqrt{\log |\F|})$ query algorithm for finite family matrix approximation in the two-sided matvec query model (\Cref{thm:robust}) is a near quadratic improvement on the best achievable bound using either one-sided matvec queries, or vector-matrix-vector queries (which are strictly weaker than one-sided queries).

In this section we first show that a bound of $O(\log|\F|/\delta)/\epsilon^2)$ is achievable in the vector-matrix-vector query model, matching the bound of \Cref{clm:one_sided_hutch} for the one-sided query model. We then show that $\Omega(\log |\F|)$ queries are necessary in the one-sided query model, and in turn the weaker vector-matrix-vector query model, even to achieve a weak approximation bound with small constant probability.

\subsection{Upper Bound for Vector-Matrix-Vector Queries}\label{sec:vmv}
We prove the following upper bound for finite family matrix approximation in the vector-matrix-vector model:
\begin{claim}
    \label{clm:vmv} 
    There is an algorithm that, for any finite matrix family $\mathcal{F}$, $\epsilon,\delta \in (0,1)$, and $\bv{A}\in \R^{n\times n}$ accessible via vector-matrix-vector queries of the form $\bv x,\bv y \rightarrow \bv x^\T \bv{A}\bv y$, issues $O(\log (|\F|/\delta)/\epsilon^2)$ such queries and, with probability $> 1-\delta$, returns $\bv{\tilde{B}}\in \mathcal{F}$ satisfying:
    \begin{align*}
    \|\bv{A} - \bv{\tilde{B}}\|_{\Fro} \leq (1+\epsilon) \cdot \min_{\bv{B}\in \family}\|\bv{A} - \bv{B}\|_{\Fro}.
    \end{align*}
    \end{claim}
    \begin{proof}
        As with \Cref{clm:one_sided_hutch}, the result follows by using a Frobenius norm approximation result like \Cref{fact:hutch_standard}, and union bounding over all matrices $\bv A - \bv B$ for $\bv B \in \F$. However, the Frobenius norm approximation method must be implementable with just vector-matrix-vector queries. We show that this is possible using a standard ``median-trick'' style argument. Consider independent $\bv {x},\bv y \in \{-1,1\}^n$ with i.i.d. Rademacher entries. We will show that, for any matrix $\bv C \in \mathbb R^{n \times n}$, $(\bv x^\T \bv C \bv y)^2$ is an unbiased estimator for $\|\bv C\|_\Fro^2$ with bounded variance. In particular, we have:
        \begin{align*}
            \E&\left[(\bv x^\T \bv C \bv y)^2\right] = \mathbb E\left [\left (\sum_{i,j}  x_i  y_j  C_{ij} \right )^2 \right ]\\
            &= \sum_{i,j} \E\left[ x_i^2 y_i^2  C_{ij}^2\right] + \sum_{i \neq k, j} \E\left[ x_i x_k y_j^2  C_{ij} C_{kj}\right] + \sum_{i, k \neq j} \E\left[ x_i^2  y_k y_j C_{ij} C_{ik}\right] + \sum_{i \neq k, j \neq \ell} \E\left[ x_i  x_k  y_j y_\ell C_{ij} C_{k\ell}\right].
        \end{align*}
        Observe that all terms except those in the first sum are $0$ since $\bv x$ and $\bv y$ are independent and $\E[x_i x_j] = 0$ whenever $i \neq j$ (and analogously $\E[ y_i y_j] = 0$). Thus, using that $\E[ x_i^2] = \E[ y_i^2] = 1$ for all $i$, we have:
        \begin{align}\label{eq:camExp}
        \E\left[(\bv x^\T \bv C \bv y)^2\right] &= \sum_{i,j} \E\left[x_i^2 y_j^2 C_{ij}^2\right] = \|\bv C\|_\Fro^2.
        \end{align}
        Further, again dropping all `cross' terms that include quantities of the form $ x_i  x_j$ or $ y_i  y_j$ for $i \neq j$, we have:
        \begin{align}\label{eq:CamVar}
            \Var[(\bv x^\T \bv C \bv y)^2] &\le \E\left[(\bv x^\T \bv C \bv y)^4\right] = \sum_{i, k, j, \ell} \E[x_i^2 x_k^2 y_j^2  y_\ell^2  C_{ij}^2  C_{k\ell}^2] = \sum_{i, k, j, \ell}  C_{ij}^2  C_{k\ell}^2 = \left (\sum_{i,j}  C_{ij}^2 \right)^2 = \|\bv C\|_\Fro^4.
        \end{align}
        We can reduce the variance of $\bv x^\T \bv C \bv y$ via averaging. Let $\bv x_1,\ldots,\bv x_m \in \{-1,1\}^n$ and $\bv y_1,\ldots,\bv y_m \in \{-1,1\}^n$ be independent vectors with i.i.d. Rademacher entries. Let $S(\bv{C}) = \frac{1}{m} \sum_{i=1}^m (\bv x_i^\T \bv C \bv y_i)^2$.
        Then, by \eqref{eq:camExp} and \eqref{eq:CamVar}:
        \begin{align*}
            \E[S(\bv{C})] = \|\bv C\|_\Fro^2\,\text{ and }\,
            \Var[S(\bv{C})] &\leq \frac{1}{m} \|\bv C\|_\Fro^4.
        \end{align*}
        If we set $m = \frac{1}{36 \epsilon^2}$, then we have $\Var[S(\bv{C})] \leq \frac{\epsilon^2}{36} \|\bv C\|_\Fro^4$. Thus, by Chebyshev's inequality, 
        \begin{align}\label{eq:camCheby}
            \Pr[|S(\bv{C}) - \|\bv C\|_\Fro^2| \geq \frac{\epsilon}{3} \|\bv C\|_\Fro^2] \leq \frac{\Var[S(\bv{C})]}{(\epsilon/3)^2 \|\bv C\|_\Fro^4} = \frac{1}{4}.
        \end{align}
        Now, let $S_1(\bv{C}),\ldots,S_t(\bv{C})$ be independent copies of $S(\bv{C})$, where $t = c\log(1/\delta)$ for a large enough universal constant $c$. By \eqref{eq:camCheby}, in expectation, at least $3t/4$ of these estimators satisfy $|S_i(\bv{C}) - \|\bv C\|_\Fro^2| < \frac{\epsilon}{3} \|\bv C\|_\Fro^2$. By a standard Chernoff bound, great than $t/2$ satisfy this bound with probability at least $1-\delta$. Thus, with probability at least $1-\delta$,
        $$(1-\epsilon/3) \|\bv C\|_\Fro^2 \le \mathrm{median}(S_1(\bv{C}),\ldots,S_t(\bv{C})) \le (1+\epsilon/3) \|\bv C\|_\Fro^2.$$
    Setting $\delta' = \delta/|\F|$ and considering all $\bv C$ of the form $\bv C = \bv A - \bv B$ for $\bv B \in \F$, we can apply a union bound to claim that, with probability at least $1-\delta$, simultaneously for all $\bv B \in \F$, we have:
    $$(1-\epsilon/3) \|\bv A - \bv  B\|_\Fro^2 \le \mathrm{median}(S_1(\bv{A}-\bv{B}),\ldots,S_t(\bv{A}-\bv{B})) \le (1+\epsilon/3) \|\bv A - \bv B\|_\Fro^2.$$
    Thus, if we set $\displaystyle \bv{\tilde B} = \arg\min_{\bv B \in \F} \mathrm{median}(S_1(\bv{A}-\bv{B}),\ldots,S_t(\bv{A}-\bv{B}))$, we have $$\|\bv A - \bv{\tilde B}\|_\Fro^2 \leq \frac{(1+\epsilon/3)}{(1-\epsilon/3)} \cdot \min_{\bv B \in \F} \|\bv A - \bv B\|_\Fro^2 \le (1+\epsilon) \cdot \min_{\bv B \in \F} \|\bv A - \bv B\|_\Fro^2,$$ with probability at least $1-\delta$, as desired.

    Computing $S_1(\bv{A}-\bv{B}),\ldots,S_t(\bv{A}-\bv{B})$ for all $\bv{B}\in \mathcal{F}$ requires just $m \cdot t = O(\log(1/\delta')/\epsilon^2) = O(\log(|\F|/\delta)/\epsilon^2)$ vector-matrix-vector queries with $\bv A$, giving the claimed query complexity.
    \end{proof}

\subsection{Lower Bound for One-Sided Queries}
\label{sec:lower_onesided}

We next show that the bounds of \Cref{clm:one_sided_hutch} and \Cref{clm:vmv} are optimal in their dependence on $\log |\F|$: in the one-sided matvec query model, and in turn the vector-matrix-vector query model, $ \Omega(\log |\F|)$ queries are necessary to find even a coarse approximation from any finite matrix family. In particular, we show: 

\begin{theorem}
\label{thm:lower_bound_onesided}
For any $n > 0$, any approximation factor $\gamma \ge 1.1$, and any $k \leq c'n \log \gamma$ for a universal constant $c'$, there 
    is a finite family of matrices $\F \subseteq \R^{n\times n}$ with $\log |\F| = O(k)$ and a distribution over  $\bv A \in \R^{n \times n}$ such that, if $\bv{\tilde B}$ is the output of any algorithm that makes fewer than $\frac{c\sqrt{\log|\F|}}{\sqrt{\log \gamma}}$ matvec queries to $\bv A$ (and none to $\bv A^\transpose$) for a universal constant $c > 0$, 
\[ \|\bv A- \bv{\tilde B}\|_\Fro \geq \gamma \cdot \min_{\bv B \in \F} \|\bv A - \bv B\|_\Fro, \]
 with probability $\geq 9/10$ (over the randomness in $\bv A$ and possibly the algorithm).
In particular, $\Omega(\log |\F|)$ one-sided matvec queries are necessary to achieve any constant approximation factor $\gamma$ with probability $\ge 1/10$.
\end{theorem}
We remark that, similar to \Cref{thm:two_sided_lower_bound}, the dependence on $1/\log \gamma$ in \Cref{thm:lower_bound_onesided} is optimal, since in \Cref{clm:hutch_coarse}, we show how to find a $\gamma = O(|\F|)$-approximate solution using just $O(1) = O(\log |\F|/\log(\gamma))$ one-sided matvecs.

The idea behind \Cref{thm:lower_bound_onesided} is simple. As in the proof of \Cref{thm:two_sided_lower_bound}, we can focus on the case when $k = c' n\log \gamma$. We let $\bv{A}$ be distributed so that its first row is a random Gaussian vector and the rest of its rows are all zero. $\F$ will consist of matrices that are also only nonzero on their first row, each corresponding to a vector from a sufficiently fine cover (see \Cref{def:cover}) for a ball in $\R^n$. The radius of the ball is chosen so that it contains the first row of $\bv A$ with high probability. The accuracy of the cover will depend inversely on $\gamma$, so $\log |\F| = O(n  \log \gamma)$.
To solve \Cref{prob:intro}, one needs to find a vector in the cover that is close to the first row of $\bv A$. However, each matvec query $\bv x \rightarrow \bv A \bv x$ only gives a single linear measurement of this row. It can thus be shown that $\Omega(n) = \Omega(\log|\F|/\log \gamma)$ such measurements are necessary to find any reasonable approximation to the row. 

Our formal proof will use two basic lemmas about Gaussian vectors, stated below. Throughout, we use the notation $\mathcal{N}(\bv x, \bv C)$ to denote a Gaussian distribution with mean $\bv x \in \R^n$ and covariance $\bv C \in \R^{n \times n}$. We  use $\bv{I}_{n}$ to denote the $n \times n$ identity matrix, dropping the subscript when the dimension is clear from context. We use $\bv{0}$ to denote the all zero vector of length $n$.

\begin{lemma}{(Gaussian anti-concentration)}
\label{lem:other_anti_concentration}
Let $\bv{g} \in \R^{n}$ be drawn from $\mathcal{N}(\bv 0,\bv I)$. For any fixed $\bv a\in \mathbb{R}^n$ and any $\epsilon >0$,
    \[ \Pr[\|\bv{g}-\bv a\|_2 \leq \epsilon] \leq \sqrt{\frac{n}{4\pi}\left(\frac{e\epsilon^2}{n}\right)^{n}} \]
\end{lemma}
\begin{proof}
    The probability density function of $\mathcal{N}(\bv 0,\bv I)$ is upper bounded by $\frac{1}{(2\pi)^{n/2}}$. Thus, $$\Pr[\|\bv{g}-\bv a\|_2 \leq \epsilon] \le \frac{1}{(2\pi)^{n/2}} \cdot \mathrm{vol}(\mathcal B(\epsilon,\bv a)),$$ where $\mathcal B(\epsilon,\bv a)$ is an $\ell_2$ ball of radius $\epsilon$ around $\bv a$. We have 
    \[\mathrm{vol}(\mathcal B(\epsilon,\bv a)) = \frac{\pi^{n/2}\epsilon^n}{\Gamma(n/2+1)} \leq \frac{\pi^{n/2}\epsilon^n}{\sqrt{\frac{2\pi}{n/2}}\left(\frac{n/2}{e}\right)^{n/2}},\]
    where the inequality follows from Stirling's approximation to the Gamma function.
    Thus, we have:
    \begin{align*}
       \Pr[\|\bv{g}-\bv a\|_2 \leq \epsilon]
       &\leq \frac{\epsilon^n}{\sqrt{\frac{4\pi}{n}}\left(\frac{n}{e}\right)^{n/2}}
       = \sqrt{\frac{n}{4\pi}\left(\frac{e\epsilon^2}{n}\right)^{n}}. \qedhere
    \end{align*}
\end{proof}

Without loss of generality we can assume that any (possibly adaptive) matvec algorithm only issues queries $\bv{q}_1, \ldots, \bv{q}_m$ that are that are unit norm and orthogonal to each other; see e.g., \cite{SunWoodruffYang:2019}. In partilcuar, the query responses for a protocol that issues arbitrary queries $\bv{x}_1, \ldots, \bv{x}_m$ can be exactly computed from the responses of a protocol that instead issues queries that from a nested orthonormal basis for $\bv{x}_1, \ldots, \bv{x}_m$ -- i.e., from a matrix $\bv{Q}\in \R^{n\times m}$ where, for any $i$, the first $i$ columns span $\bv{x}_1, \ldots, \bv{x}_i$. Such a basis could be obtained by orthogonalizing $\bv{x}_i$ on-the-fly against all prior queries. 

With this simplifying assumption in mind, and recalling that our target matrix $\bv{A}$ will be chosen to be non-zero only in its first row, our problem will thus reduce to approximating that row (which is chosen to be a random Gaussian vector) given access to the inner products with a set of orthonormal queries. We make the following claim about the distribution of the row given the outcomes of any adaptively chosen set of orthonormal queries:
\begin{lemma}
\label{lem:one_sided_hidden_gaussian}
Let $\bv{g} \in \R^n$ be drawn from $\mathcal{N}(\bv 0,\bv I)$. For any $m \leq n$, consider any set of adaptively chosen, unit norm and mutually orthogonal query vectors $\bv{q}_1, \ldots, \bv{q}_m$, where $\bv{q}_i$ is drawn from a distribution $\mathcal{D}_i$ that is a function of $\bv{q}_1, \bv{q}_1^\transpose \bv{g}, \ldots, \bv{q}_{i-1}, \bv{q}_{i-1}^\transpose \bv{g}$.
Let $\bv{Q}_m \in \R^{n \times m}$ have $\bv q_1,\ldots,\bv q_m$ as its columns and let $\bv{g}_m = \bv{g} - \bv{Q}_m\bv{Q}_m^\transpose \bv{g}$.
The distribution of $\bv{g}_m$ conditioned on the outcomes of 
$\bv{q}_1, \bv{q}_1^\transpose \bv{g}, \ldots, \bv{q}_{m}, \bv{q}_{m}^\transpose \bv{g}$ is $\mathcal{N}(\bv{0}, \bv{I} - \bv{Q}_m\bv{Q}_m^\transpose)$.
\end{lemma}
\begin{proof}
The proof is inductive. For $i \leq m$, define $\bv{Q}_i := [\bv{q}_1, \ldots, \bv{q}_i]$ and $\bv{g}_i := (\bv{I} - \bv{Q}_i\bv{Q}_i^\transpose)\bv{g}$. 
Inductively assume that the distribution of $\bv{g}_i$ conditioned on 
$\{\bv{q}_1, \bv{q}_1^\transpose \bv{g}, \ldots, \bv{q}_{i}, \bv{q}_{i}^\transpose \bv{g}\}$ is $\mathcal{N}(\bv{0}, \bv{I} - \bv{Q}_i\bv{Q}_i^\transpose)$.
Since the distribution of $\bv{q}_{i+1}$ is a function of $\{\bv{q}_1, \bv{q}_1^\transpose \bv{g}, \ldots, \bv{q}_{i}, \bv{q}_{i}^\transpose \bv{g}\}$, it follows that $\bv{g}_i$ is still distributed as $\mathcal{N}(\bv{0}, \bv{I} - \bv{Q}_i\bv{Q}_i^\transpose)$ when conditioned on $\{\bv{q}_1, \bv{q}_1^\transpose \bv{g}, \ldots, \bv{q}_{i}, \bv{q}_{i}^\transpose \bv{g}, \bv{q}_{i+1}\}$. 
By rotational invariance of the Gaussian distribution, we can thus write:
\begin{align*}
\bv{g}_i = z\cdot \bv{q}_{i+1} + (\bv{I} - \bv{Q}_{i+1}\bv{Q}_{i+1}^\transpose)\bv{h}, 
\end{align*}
where $\bv{h} \sim\mathcal{N}(\bv{0}, \bv{I})$, $z \sim \mathcal{N}(0,1)$, and $\bv{h}, z$ are independent from each other and from $\{\bv q_1, \bv q_1^\transpose \bv g, \ldots, \bv q_i, \bv q_i^\transpose \bv g, \bv q_{i+1}\}$.
Observing that $\bv{g}_{i+1} = (\bv{I}-\bv{q}_{i+1}\bv{q}_{i+1}^\transpose)\bv{g}_i$, we thus that, conditioned on $\{\bv{q}_1, \bv{q}_1^\transpose \bv{g}, \ldots, \bv{q}_{i}, \bv{q}_{i}^\transpose \bv{g}, \bv{q}_{i+1}\}$,
\begin{align*}
\bv{g}_{i+1} &= (\bv I - \bv Q_{i+1} \bv Q_{i+1}^\T)\bv{h} & &\text{and} & \bv{q}_{i+1}^T\bv{g} &= \bv{q}_{i+1}^T\bv{g}_{i+1} = z.
\end{align*}
Since $\bv{h}$ and $z$ are conditionally independent, we conclude that $\bv{g}_{i+1}$ is still distributed as $(\bv{I} - \bv{Q}_{i+1}\bv{Q}_{i+1}^\transpose) \bv{h}$ even conditioned on all $i+1$ queries and responses, $\{\bv{q}_1, \bv{q}_1^\transpose \bv{g}, \ldots, \bv{q}_{i}, \bv{q}_{i}^\transpose \bv{g}, \bv{q}_{i+1}, \bv{q}_{i+1}^\T\bv{g}\}$, as required.

For the base case, $\bv{q}_1$ is independent of $\bv{g}$ so, $(\bv{I} -\bv{q}_1\bv{q}_1^\transpose)\bv{g}$ is distributed as $\mathcal{N}(\bv{0}, \bv{I} - \bv{q}_1\bv{q}_1^\transpose)$. 
\end{proof}

Using \Cref{lem:other_anti_concentration,lem:one_sided_hidden_gaussian}, we can now prove \Cref{thm:lower_bound_onesided}.

\begin{proof}[Proof of \Cref{thm:lower_bound_onesided}]
    First note that it suffices to prove the lower bound for  $\log |\F| = O(n \log \gamma)$. To obtain a lower bound for smaller choices of $\F$, we simply embed the hard instance described below in the top left $n\times n$ block of a larger matrix, setting all values not in that block to zero, both in $\bv{A}$ and in the matrices in $\mathcal{F}$. 

As discussed, we choose $\bv{A} \in \R^{n \times n}$ to have its first row distributed as $\mathcal{N}(\bv 0, \bv I)$ and all other rows equal to $\bv{0}$
 --- i.e., $\bv A = \bv{e}_1 \bv g^\transpose$ for $\bv g\sim \mathcal{N}(\bv 0, \bv I)$, where $\bv{e}_1$ is the first standard basis vector.
 For constants $c_1$ and $c_2$ to be set later, let $\alpha = \frac{c_1 \sqrt{n}}{\gamma}$ and let  $\mathcal{C}_\alpha$ be an $\alpha$-cover (see \Cref{def:cover}) of the ball $\mathcal{B}(\sqrt{c_2 n}) = \{\bv x \in \R^n : \|\bv x\|_2 \le \sqrt{c_2 n}\}$. Let $\F \subset \R^{n\times n}$ be the set of matrices generated by taking an element of $\mathcal C_\alpha$ and placing it in the first row, with all other rows set to $0$. I.e., $\F = \{\bv{e}_1 \bv x^\T : \bv x \in \mathcal C_\alpha\}$. It is well known (c.f. \cite[Corollary 4.2.13]{Vershynin2018}) that we can construct $\mathcal C_\alpha$ with $|\mathcal C_\alpha | = \left (\frac{2 \sqrt{c_2 n}}{\alpha} +1\right )^n = \exp(O(n \cdot \log\gamma))$ and thus $\log|\F| = \log|\mathcal C_\alpha| = O(n \cdot \log \gamma)$. Note that here we use the assumption that $\gamma > 1.1$ to ensure that $\log(2 \gamma\sqrt{c_2}/c_1+1) = O(\log \gamma)$.
 
\smallskip

\noindent \textbf{Upper bound on the optimal error.}
By a standard Chi-squared concentration  bound, if $c_2$ is set large enough, then $\bv g \in \mathcal{B}(\sqrt{c_2 n})$ with probability at least $19/20$. 
Assume that this event holds. Then by the definition of an $\alpha$-cover:
\begin{align}\label{eq:camMinBound}
    \min_{\bv B \in \F} \|\bv A - \bv B\|_\Fro \le \alpha = \frac{c_1 \sqrt{n}}{\gamma}.
\end{align}

\smallskip

\noindent \textbf{Lower bound on algorithm's error.}
By \Cref{lem:one_sided_hidden_gaussian}, the  distribution of $\bv{A}$ conditioned on any $m$ adaptively chosen matvec queries of the form $\bv x \rightarrow \bv A \bv x$ is $\bv e_1 \bv g_m^\T$, where $\bv g_m$ is drawn from $\mathcal{N}(\bv 0,\bv I - \bv Q_m \bv Q_m^\T)$ and $\bv Q_m \in \R^{n \times m}$ is an orthonormal basis for the query vectors. That is, conditioned on the outcomes of the queries, $\bv g_m$ is distributed as $\bv Z \bv h$ where $\bv Z \in \R^{n \times (n-m)}$ is an orthonormal basis for which $\bv Z \bv Z^\T = \bv {I} - \bv Q_m \bv Q_m^\T$ and $\bv h$ is distributed as $\mathcal{N}(\bv 0,\bv I_{n-m})$.
For any fixed vector $\bv a \in \R^{n}$ we have: 
\begin{align*}
    \Pr[\|\bv g_m - \bv a\|_2 \leq c_1 \sqrt{n}] &\le \Pr[\|\bv Z^\T \bv g_m -\bv Z^\T\bv a\|_2 \leq c_1 \sqrt{n}] = \Pr[\|\bv h -\bv Z^\T\bv a\|_2 \leq c_1 \sqrt{n}]
\end{align*}
where the first inequality uses that $\bv {Z}$ has orthonormal columns and so $\|\bv Z^\T \bv g_m -\bv Z^\T\bv a\|_2  \le \|\bv g_m -\bv a\|_2 $.
Invoking \Cref{lem:other_anti_concentration} with $\epsilon = c_1 \sqrt{n}$ and choosing $c$ small enough that $m \le \frac{c \log |\mathcal F|}{\log \gamma} \le \frac{n}{2}$, we have that, for sufficiently small $c_1$,
\begin{align*}
\Pr\left[\|\bv g_m - \bv a\|_2 \leq c_1 \sqrt{n}\right] \leq \sqrt{\frac{n-m}{4\pi}\left(\frac{e c_1^2 n}{n-m}\right)^{n-m}}
\leq \sqrt{\frac{n-m}{4\pi}\left(2e c_1^2\right)^{n-m}}
\leq \frac{1}{20}.
\end{align*}

Since the output $\bv{\tilde B}$ of any algorithm only depends on the $m$ queries and responses, we thus have that $\|\bv A - \bv{\tilde B}\|_\Fro \geq c_1 \sqrt{n}$ with probability $\geq 19/20$. Combined with \eqref{eq:camMinBound} and a union bound over both bounds holding, we have that, as desired, with probability at least $9/10$,
$\|\bv A - \bv{\tilde B}\|_\Fro \geq \gamma \cdot  \min_{\bv B \in \F} \|\bv A - \bv B\|_\Fro$.
\end{proof}

\end{document}

%% file: bib_macros.tex
  \usepackage{nth}
  \usepackage{intcalc}

  \newcommand{\cAAAI}[1]{AAAI\ Conference\ on\ Artificial (AAAI)}